\documentclass[runningheads]{llncs}
\usepackage{ae,aecompl}

\usepackage{amsmath,amssymb}
\usepackage[final]{graphicx}
\usepackage{verbatim,wasysym}
\usepackage{wrapfig}

\usepackage[ruled]{algorithm}
\usepackage{algorithmicx}
\usepackage[noend]{algpseudocode}
\usepackage{latexsym}

\newcommand{\proc}[1]{\textnormal{\scshape#1}}

\newcommand{\kw}[1]{{\textbf{#1}}} 
\newcommand{\To}{{\kw{to}}} 

\newcommand{\func}[1]{%
  \ensuremath{\mathop{#1}\nolimits}}
\providecommand{\abs}[1]{\ensuremath{\lvert #1 \rvert}}
\providecommand{\nat}[0]{\ensuremath{\mathbb N}}
\providecommand{\seq}[3]{\ensuremath{#1_{#2}, \dotsc, #1_{#3}}}

\spnewtheorem{clm}{Claim}{\itshape}{\rmfamily}

\DeclareMathOperator{\lca}{lca}
\DeclareMathOperator{\level}{level}
\DeclareMathOperator{\father}{f}
\DeclareMathOperator{\values}{{value}}
\DeclareMathOperator{\state}{{state}}
\DeclareMathOperator{\rank}{{rank}}
\DeclareMathOperator{\sig}{{sig}}
\renewcommand{\O}{{\mathcal{O}}}

\newcommand{\frowny}{{\ensuremath{\frownie}}}

\DeclareMathOperator{\Ker}{Ker}
\DeclareMathOperator{\minlevel}{min-level}

\newcommand{\kminimise}{\proc{$k$-Minimise}}
\newcommand{\distancetree}{\proc{Distance-Tree}}

\newcommand{\inlevel}{{\ensuremath{\mathop{in\text{-}level}\nolimits}}}
\newcommand{\makeset}[2]{\ensuremath{ \{ #1 \: | \: #2 \} }}

\newcommand{\krowne}[1][k]{{\sim_{#1}}}
\newcommand{\knierowne}[1][k]{{\not \sim_{#1}}}

\newcommand{\automat}[1][q_0]{{\langle Q, \Sigma , \delta, #1 , F \rangle}}

\providecommand{\simdiff}[0]{\ensuremath{\mathop{\bigtriangleup}}}

\usepackage{color}

\definecolor{myBlue}{rgb}{0.5,0.5,1}
\definecolor{myRed}{rgb}{0.9,0.3,0.1}
\definecolor{myYellow}{rgb}{0.9,0.9,0}
\definecolor{myGreen}{rgb}{0.1,1,0}

\newcommand{\artur}[1]{\noindent\colorbox{myYellow}{Artur: #1}}

\allowdisplaybreaks

\title{On minimising automata with errors
    }

\author{Pawe\l{} Gawrychowski \inst{1} \fnmsep \thanks{Supported by
    MNiSW grant number N~N206 492638, 2010--2012.} \and Artur Je\.z
  \inst{1} \fnmsep ${}^{\star}$ \and Andreas Maletti \inst{2}
  \fnmsep \thanks{Supported by the \emph{Ministerio de Educaci\'on y
      Ciencia} (MEC) grant JDCI-2007-760 and the German Research
    Foundation~(DFG) grant MA/4959/1-1.}}
\authorrunning{P.~Gawrychowski \and A.~Je\.z \and A.~Maletti}
\institute{Institute of Computer Science, University of Wroc{\l}aw \\
  ul.\ Joliot-Curie~15, 50-383 Wroc{\l}aw, Poland \\
  \email{\{gawry,aje\}@cs.uni.wroc.pl} \and
  Institute for Natural Language Processing, Universit\"at Stuttgart \\
  Azenbergstra\ss e~12, 70174 Stuttgart, Germany \\
  \email{andreas.maletti@ims.uni-stuttgart.de}}

\begin{document}
 
\maketitle
\begin{abstract}
  The problem of $k$-minimisation for a DFA~$M$ is the computation of
  a smallest DFA~$N$ (where the size~$\abs M$ of a DFA~$M$ is the size
  of the domain of the transition function) such that $L(M) \simdiff
  L(N) \subseteq \Sigma^{<k}$, which means that their recognized
  languages differ only on words of length less than~$k$.  The
  previously best algorithm, which runs in time~$\mathcal{O}(\abs M \log^2 n)$
  where $n$~is the number of states, is extended to DFAs with partial
  transition functions.  Moreover, a faster $\mathcal{O}(\abs M \log
  n)$~algorithm for DFAs that recognise finite languages is presented.
  In comparison to the previous algorithm for total DFAs, the new
  algorithm is much simpler and allows the calculation of
  a~$k$-minimal DFA for each~$k$ in parallel.  Secondly, it is
  demonstrated that calculating the least number of introduced errors
  is hard: Given a DFA~$M$ and numbers $k$~and~$m$, it is
  \mbox{NP-hard} to decide whether there exists a \mbox{$k$-minimal}
  DFA~$N$ with $\abs{L(M) \simdiff L(N)} \leq m$.  A~similar result
  holds for hyper-minimisation of DFAs in general:  Given a DFA~$M$ and
  numbers $s$~and~$m$, it is NP-hard to decide whether there exists a
  DFA~$N$ with at most $s$~states such that $\abs{L(M) \simdiff L(N)}
  \leq m$.

\end{abstract}


\section{Introduction}
\label{sec:Intro}
Deterministic finite automata~(DFAs) 
are one of the simplest devices recognising languages.  The study of
their properties is motivated by (i)~their simplicity, which yields
efficient operations, (ii)~their wide-spread applications, (iii)~their
connections to various other areas in theoretical computer science,
and (iv)~the apparent beauty of their theory.  A DFA~$M$ is a
quintuple~$\automat$, where $Q$~is its finite state-set, $\Sigma$~is
its finite alphabet, $\delta \colon Q \times \Sigma \to Q$ is its
partial transition function, $q_0 \in Q$ is its starting state, and $F
\subseteq Q$ is its set of accepting states.  The DFA~$M$ is total if
$\delta$ is total.  The transition function~$\delta$ is extended to
$\delta \colon Q \times \Sigma^* \to Q$ in the standard
way.  
The language~$L(M)$
that is \emph{recognised by the DFA}~$M$ is $L(M) =
\makeset{w}{\delta(q_0, w) \in F}$.

Two DFAs $M$~and~$N$ are \emph{equivalent} (written as~$M \equiv N$) if
$L(M) = L(N)$.  A DFA~$M$ is \emph{minimal} if all equivalent DFAs are
larger.  One of the classical DFA~problems is the \emph{minimisation
  problem}, which given a DFA~$M$ asks for the (unique) minimal
equivalent DFA.  The asymptotically fastest DFA minimisation algorithm
runs in time~$\mathcal{O}(\abs \Sigma\, n \log n)$ and is due to
\textsc{Hopcroft}~\cite{DFAminimisation,DFAGries}, where $n = \abs Q$;
its variant for partial DFAs is known to run in time~$\mathcal{O}(\abs M \log
n)$.


Recently, minimisation was also considered for
hyper-equivalence~\cite{Badr,BadrRAIRO}, which allows a finite
difference in the languages.  Two languages $L$~and~$L'$ are
\emph{hyper-equivalent} if $\abs{L \simdiff L'} < \infty$, where
$\simdiff$~denotes the symmetric difference of two sets.  The DFAs
$M$~and~$N$ are hyper-equivalent if their recognised languages are.
The DFA~$M$ is \emph{hyper-minimal} if all hyper-equivalent DFAs are
larger.  The algorithms for hyper-minimisation~\cite{BadrRAIRO,Badr}
were gradually improved over time to the currently best run-time
$\mathcal{O}(\abs M \log^2 n)$~\cite{HolzerMalettifull,GawrychowskiJezMFCS},
which can be reduced to $\mathcal{O}(\abs M \log n)$ using a~strong
computational model (with randomisation or special memory access).
Since classical DFA minimisation linearly reduces to
hyper-minimisation~\cite{HolzerMalettifull}, an algorithm that is
faster than $\mathcal{O}(\abs M \log n)$ seems unlikely.  Moreover, according
to the authors' knowledge, randomisation does not help
\textsc{Hopcroft}'s~\cite{Restivo} or any other
DFA~minimisation algorithm.  Thus, the randomised hyper-minimisation
algorithm also seems to be hard to improve.

Already~\cite{BadrRAIRO} introduces a stricter notion of
hyper-equivalence. Two languages $L$~and~$L'$ are $k$-\emph{similar}
if they only differ on words of length less than~$k$.  Analogously, DFAs
are $k$-similar if their recognised languages are.  A DFA~$M$ is
\mbox{\emph{$k$-minimal}} if all $k$-similar DFAs are larger, and the
\emph{$k$-minimisation problem} asks for a $k$-minimal DFA that is
$k$-similar to the given DFA~$M$.  The known
algorithm~\cite{GawrychowskiJezMFCS} for $k$-minimisation of total
DFAs runs in time~$\mathcal{O}(\abs M \log^2 n)$, however it is quite
complicated and fails for non-total DFAs.

In this contribution, we present a simpler $k$-minimisation algorithm
for general DFAs, which still runs in time~$\mathcal{O}(\abs M \log^2 n)$.
This represents a significant improvement compared to the complexity
for the corresponding total DFA if the transition table of~$M$ is
sparse.  Its running time can be reduced if we allow a stronger
computational model.  In addition, the new algorithm runs in
time~$\mathcal{O}(\abs M \log n)$ for every DFA~$M$ that recognises a finite
language.  Finally, the new algorithm can calculate (a compact
representation of) a $k$-minimal DFA for each possible~$k$ in a
single run (in the aforementioned run-time).  Outputting all the
resulting DFAs might take time $\Omega(n \abs M \log ^2 n)$.

Although $k$-minimisation can be efficiently performed, no uniform
bound on the number of introduced errors is provided.  In the case of
hyper-minimisation, it is known~\cite{mal10c} that the \emph{optimal}
(i.e., the DFA committing the least number of errors) hyper-minimal
DFA and the number of its errors~$m$ can be efficiently computed.
However, this approach does not generalise to $k$-minimisation.
We show that this is for a reason: already the problem of calculating the number~$m$ of
errors of an optimal $k$-minimal automaton is \mbox{NP-hard}.

Finally, for some applications it would be beneficial if we could
balance the number~$m$ of errors against the compression
rate~$\frac{\abs N}{\abs M}$.  Thus, we also consider the question
whether given a DFA~$M$ and two integers $s$~and~$m$ there is a
DFA~$N$ with at most $s$~states that commits at most $m$~errors (i.e.,
$\abs{L(M) \simdiff L(N)} \leq m$).  Unfortunately, we show that this
problem is also \mbox{NP-hard}.


\section{Preliminaries}
\label{sec:prel}
We usually use the two DFAs $M = \automat$ and $N =
\langle P, \Sigma, \mu, p_0, F'\rangle$.  We also write~$\delta(w)$
for~$\delta(q_0,w)$.  The \emph{right-language} $L_M(q)$ of a state~$q
\in Q$ is the language $L_M(q) = \makeset{w}{\delta(q, w) \in F}$
recognised by~$M$ starting in state~$q$.  Minimisation of DFAs is
based on calculating the equivalence~$\equiv$ between states, which is
defined by $q \equiv p$ if and only if $L_M(q) = L_N(p)$.  Similarly,
the \emph{left language} of~$q$ is the language
$\delta^{-1}(q) = \makeset{w}{\delta(w) = q}$ of words leading to~$q$
in~$M$.

For two languages $L$~and~$L'$, we define their \emph{distance}~$d(L,
L')$ as
\[ d(L,L') = \min{} \makeset{\ell}{L \cap \Sigma^{\geq \ell} = L' \cap
  \Sigma^{\geq \ell}} \enspace, \] where $\min{} \emptyset = \infty$.
Actually, $d$~is an ultrametric.
The distance~$d$
can be extended to states: $d(q, p) = d(L_M(q), L_N(p))$ for $q \in Q$
and $p \in P$.  It satisfies the simple recursive formula:
\begin{equation}
  \label{eq:distance recursion}
  d(q, p) = \begin{cases} 0 & \text{ if } q \equiv p,\\
    1 + \max{} \makeset{d(\delta(q, a), \mu(p, a))}{a \in
      \Sigma} & \text{ otherwise}. \end{cases}
\end{equation}
Since $d$~is an ultrametric on languages, \eqref{eq:distance
  recursion} yields that the distance~$d(q_1, q_2)$ between $q_1, q_2
\in Q$ in the DFA~$M$ is either infinite or small.  Formally, $d(q_1,
q_2) = \infty$ or $d(q_1, q_2) < \abs Q$.

The minimal DFAs considered in this paper are obtained mostly by state
merging.  We say that the DFA~$N$ is the result of \emph{merging
  state~$q$ to state~$p$} (assuming $q \neq p$) in~$M$ if $N$~is
obtained from~$M$ by changing all transitions ending in~$q$ to
transitions ending in~$p$ and deleting the state~$q$.  If $q$~was the
starting state, then $p$~is the new starting state.  Formally, $P = Q
\setminus \{q\}$, $F' = F \setminus \{q\}$, and
\begin{align*}
  \mu(r, a) &=
  \begin{cases}
    p & \text{if } \delta(r, a) = q \\
    \delta(r, a) & \text{otherwise,}
  \end{cases} &
  p_0 &=
  \begin{cases}
    p & \text{if } q_0 = q \\
    q_0 & \text{otherwise.}
  \end{cases}
\end{align*}
The process is illustrated in Fig.~\ref{fig:merge}.

\begin{figure}[h]
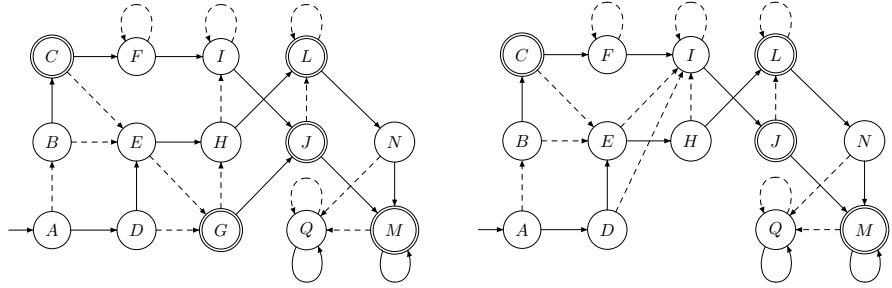

  \begin{center}
    \includegraphics[width=0.45 \textwidth]{merge3.mps} \qquad 
    \includegraphics[width=0.45 \textwidth]{merge4.mps}
  \end{center}
  \caption{Merging state~$G$ into~$I$.}
  \label{fig:merge}
\end{figure}

Finally, let~\emph{$\inlevel_M(q)$} be the length of the longest word
leading to~$q$ in~$M$.  If there is no such longest word, then
$\inlevel_M(q) = \infty$.  Formally, $\inlevel_M(q) = \sup{}
\makeset{\abs w} {w \in \delta^{-1}(q)}$ for every $q \in Q$.  The
structural characterisation of hyper-minimal
DFAs~\cite[Sect.~3.2]{BadrRAIRO} relies on a state classification into
\emph{kernel} and \emph{preamble} states.  The set~$\Ker(M)$ of kernel
states consists of all states~$q \in Q$ with $\inlevel_M(q) = \infty$,
whereas the remaining states are preamble states.  Roughly speaking,
the kernels of two hyper-equivalent and hyper-minimal automata are
isomorphic in the standard sense, and their preambles are also
isomorphic except for acceptance values.

\section{Efficient $k$-minimisation}
\label{k-minimisation}
\subsection{$k$-similarity and $k$-minimisation}
\label{k-minimisation2}
Two languages $L$~and~$L'$ are \emph{$k$-similar} if they only differ
on words of length smaller than~$k$, and the two DFAs $M$~and~$N$ are
$k$-similar if their recognised languages are.  The DFA~$M$ is
\emph{$k$-minimal} if all $k$-similar DFAs are larger.  In this
section, we first give a general simple algorithm~\kminimise{} that
computes a~$k$-minimal DFA that is $k$-similar to the input DFA~$M$.
Then we present a data structure that allows a fast, yet simple
implementation of this algorithm.

\begin{definition}
  \label{df:kEQ}
  For two languages $L$~and~$L'$, we let
  \( L \sim_k L' \iff d(L, L') \leq k  \).
\end{definition}

The hyper-equivalence relation~\cite{BadrRAIRO} can be now defined as
$\mathord{\sim} = \bigcup_k \mathord{\sim_k}$.  
Next, we extend $k$-similarity to states.

\begin{definition}
  \label{df:kEQstate}
  Two states $q \in Q$ and $p \in P$ are $k$-similar, denoted by
  $q \sim_k p$, if 
  \[ d(q, p) + \min(k, \inlevel_M(q), \inlevel_N(p)) \leq k
  \enspace. \]
\end{definition}

While $\sim_k$ is an equivalence relation on languages, it is, in
general, only~a compatibility relation (i.e., reflexive and symmetric)
on states.  On states the hyper-equivalence is not a direct
generalisation of $k$-similarity.  Instead, $p \sim q$ if and only if
$L_M(q) \sim L_N(p)$.  We use the $k$-similarity relation to give~%
a simple algorithm $\kminimise(M)$, which constructs a
$k$-minimal DFA (see Algorithm~\ref{alg:kMin}).
In Section~\ref{subsec:distance forests}
we show how to implement it efficiently.

\begin{algorithm}
  \caption{$\kminimise(M)$ with minimal~$M$}
  \label{alg:kMin}
  \begin{algorithmic}[1]
    \State calculate $\krowne$ on~$Q$
    \State $N \gets M$
    \While{$q \sim_k p$ for some $q, p \in P$ and $q \neq p$}
      \If{$\inlevel_M(q) \geq \inlevel_M(p)$}
        \State swap $q$~and~$p$
      \EndIf
    \State $N \gets \proc{Merge}(N, q, p)$
    \EndWhile
  \end{algorithmic}
\end{algorithm}

\begin{theorem}
  \label{thm:kminimise is proper}
  \kminimise{} returns a $k$-minimal DFA that is
  $k$-similar to~$M$.
\end{theorem}



  

\subsection{Distance forests}
\label{subsec:distance forests}
In this section we define distance forests, which capture the
information of the distance between states of a given minimal DFA~$M$.
We show that $k$-minimisation can be performed in linear time, when a
distance forest for~$M$ is supplied.  We start with a total DFA~$M$
because in this case the construction is fairly easy.  In
Section~\ref{subsec:partial transition function} we show how to extend
the construction to non-total DFAs.

Let $\mathcal F$~be a forest (i.e., set of trees) whose leaves are
enumerated by~$Q$ and whose edges are weighted by elements of~$\nat$.
For convenience, we identify the leaf vertices with their label.  For
every $q \in Q$, we let $\mathrm{tree}(q) \in \mathcal F$ be the
(unique) tree that contains~$q$.  The level~$\level(v)$ of a
vertex~$v$ in~$t \in \mathcal F$ is the maximal weight of all paths
from~$v$ to a leaf, where the weights are added along a path.
Finally, given two vertices $v_1, v_2$ of the same tree~$t \in
\mathcal F$, the lowest common ancestor of $v_1$~and~$v_2$ is the
vertex~$\lca(v_1, v_2)$.

\begin{definition}[Distance forest]
  Let $\mathcal F$ be a forest whose leaves are enumerated by~$Q$.
  Then $\mathcal F$ is a \emph{distance forest} for~$M$ if for every
  $q, p \in Q$ we have
  \[ d(q, p) = \begin{cases}
    \level(\lca(q, p)) & \text{ if } \mathrm{tree}(q) =
    \mathrm{tree}(p),\\
    \infty & \text{ otherwise}.
  \end{cases} \]
\end{definition}


\begin{figure}
  \begin{center}
    \includegraphics{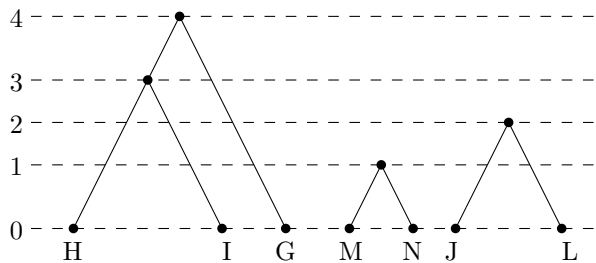} 
  \end{center}
  \caption{A distance forest for the left DFA
    of~\protect{Fig.~\ref{fig:merge}}.  Single-node trees
    are omitted.}
  \label{fig:forest}
\end{figure}
In order to construct a distance forest we use~\eqref{eq:distance
  recursion} to calculate the distance.  Mind that~$M$ is minimal, so
there are no states with distance~$0$.  In phase~$\ell$, we merge all
states at distance exactly~$\ell$ into one state.  Since we merged all
states of distance at most~$\ell - 1$ in the previous phases, we only
need to identify the states of distance~$1$ in the merged DFA.  Thus
we simply group the states according to their vectors of transitions
by letters from $\Sigma = \{\seq a1m\}$.  To this end we store these
vectors in a dictionary, organised as a trie of depth $m$.
The leaf of a trie corresponding to a~path~$(\seq q1m)$
keeps a list of all states~$q$ such that $\delta(q,
a_i) = q_i$ for every $1 \leq i \leq m$.  For each node~$v$ in the
trie we keep a \emph{linear dictionary} that maps a~state~$q$ into a
child of~$v$.  We demand that this linear dictionary supports search,
insertion, deletion, and enumeration of all elements.

\begin{theorem}
  \label{thm:linear dictionary log n times}
  Given a total DFA~$M$, we can build a distance forest for~$M$ using
  $\mathcal{O}(\abs M \log n)$ linear-dictionary operations.
\end{theorem}

\label{subsec:linear dictionary}
We now shortly discuss some possible implementations of the linear
dictionary.  An implementation using balanced trees would have linear
space consumption and the essential operations would run in
time~$\mathcal{O}(\log n)$.  If we allow randomisation, then we can use dynamic
hashing.  It has a worst-case constant time look-up and an amortised
expected constant time for updates~\cite{pagh}.
Since it is natural to assume that $\log n$~is proportional to the
size of a machine word, we can hash in constant time.  We can obtain
even better time bounds by turning to more powerful models.  In the
RAM model, we can use exponential search trees~\cite{andersson}, whose
time per operation is $\mathcal{O}(\frac{(\log\log n)^2}{\log\log\log n})$ in
linear space.  Finally, if we allow a quadratic space consumption,
which is still possible in sub-quadratic time, then we can allocate
(but not initialise) a table of size~$\abs M \times n$.  Standard
methods can be used to keep track of the actually used table entries,
so that we obtain a constant run-time for each operation, but at the
expense of $\Theta(\abs M \, n)$~space; i.e., quadratic memory
consumption.

\label{subsec:kminimisation and distance forest}
We can now use a distance forest to efficiently implement
\kminimise{}.  For each state $q$ we locate its highest ancestor~$v_q$
with $\level(v_q) \leq k - \inlevel(q)$. Then $q$~can be merged into
any state that occurs in the subtree rooted in~$v_q$ (assuming it has
a smaller~$\inlevel$).  This can be done using a depth-first traversal
on the trees of the distance forest.  A~more elaborate construction
based on this approach yields the following.
%
%
\begin{theorem}
  \label{thm:k minimisation implementation}
  Given a distance forest for~$M$, we can compute the size of a
  \mbox{$k$-minimal} DFA that is $k$-similar to~$M$ for all~$k$ in
  time $\mathcal{O}(\abs M)$.  For a fixed~$k$, we can also compute a
  $k$-minimal DFA in time $\mathcal{O}(\abs M)$.  Finally, we can run the
  algorithm in time~$\mathcal{O}(\abs M \log n )$ such that it has a
  $k$-minimal DFA stored in memory in its $k$-th phase.
\end{theorem}

\subsection{Finite languages and partial transition functions}
\label{subsec:partial transition function}
\label{subsec:finite languages}
The construction of a distance forest was based on a total transition
function~$\delta$, and the run-time was bounded by the size
of~$\delta$.  We now show a modification for the non-total case.  The
main obstacle is the construction of a distance forest for an acyclic
DFA.  The remaining changes are relatively straightforward.
\begin{theorem}\label{thm:finite distance}
  For every acyclic DFA~$M$ we can build a distance forest in time
  $\mathcal{O}(\abs M \log n)$.
\end{theorem}
\begin{proof}[sketch]
  Since $L(M)$~is finite, we have that $m(p) = \max{} \makeset{\abs w}
  {w \in L_M(p)}$ is a natural number for every state~$p$.  Let $Q_i =
  \makeset{p}{m(p) = i}$ and $Q_{<\infty} = \bigcup_i Q_i$.  Every
  state has a finite right-language, and thus every distance forest
  consists of a single tree.  We iteratively construct the fragments
  of this tree by starting from a single leaf~$\bot$, which represents
  the empty language and ``undefinedness'' of the transition function.
  Before we start to process~$Q_t$, we have already constructed the
  distance tree for~$\bigcup_{i<t} Q_i$.  The constructed fragments
  are connected to a~single path, called the \emph{spine}, which ends
  at the leaf~$\bot$ (see Fig.~\ref{fig:trie}).

  Let $Q_t = \{\seq p1s\}$, and let $v \in Q_t$.  Moreover, let
  $\father(v)$ be the vector of states ${\bf v} = (\delta(v, a) )_{a
    \in \Sigma}$, where the coordinates are sorted by a fixed order
  on~$\Sigma$.  Define the distance between those vectors as
  \begin{equation*}
    d((p_a)_{a \in \Sigma}, (p'_a)_{a \in \Sigma}) = \max{}
    \makeset{d(p_i, p'_i) + 1} {a \in \Sigma} \enspace,
  \end{equation*}
  where we know that $d(p_i, \bot) = m(p_i)$ and $d(\bot, p'_i) =
  m(p'_i)$.
  Similarly to the distance, we can define the father~$\father({\bf
    v})$ of a vector ${\bf v} = (p_a)_{a \in \Sigma}$ as $\father({\bf
    v}) = (\father(p_a))_{a \in \Sigma}$.  Then
  \begin{equation*}
    \father^{\ell+1}(v) = \father^{\ell+1}(v')
    \iff \father^{\ell}({\bf v}) = \father^{\ell}({\bf v'}).
  \end{equation*}
  We can now use a divide-and-conquer approach: First, for each vector
  we calculate its $2^k$-th ancestor, where $k = \lceil \log s /2
  \rceil$.  Then all such vectors are sorted according to their
  ancestors, in particular they are partitioned into blocks with the
  same ancestors.  After that we recurse onto those (bottom) blocks
  that have more than two entries and onto the upper block, which
  consists of the different $2^k$-ancestors.  The recursion ends for
  blocks containing at most two vectors, for which we calculate the
  distance tree directly. \qed
\end{proof}
\vspace{-0.5cm}
\begin{figure}
  \begin{center}
    \begin{minipage}[t]{0.5\linewidth}
		\includegraphics[scale=1.1]{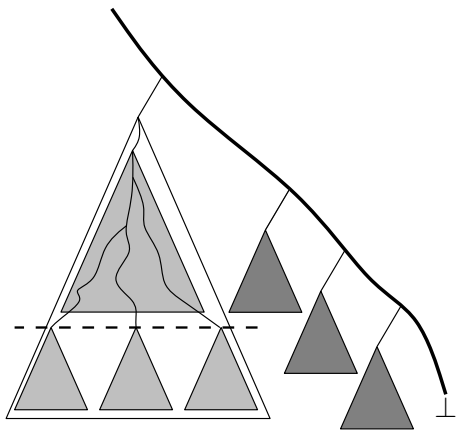} 
    \end{minipage}\hfill
    \begin{minipage}[t]{0.46\linewidth}
\vspace{-4cm}
	\caption{Illustration for the construction of the distance
          tree.  The spine is depicted using with a thicker
          line. Splitting one fragment into smaller recursive calls is
          shown.}
        \label{fig:trie}
    \end{minipage}
  \end{center}
\end{figure}
\vspace{-0.8cm}

For every state~$q \in Q$, its \emph{signature}~$\sig(q)$ is
$\makeset{a} {L_M(\delta(q, a)) \text{ is infinite}}$.  If $\sig(q)
\neq \sig(p)$, then $d(q, p) = \infty$, which allows us to keep a
separate dictionary for each signature.  Let us fix such a trie. To
take into account also the transitions by letters outside the
signature, we introduce a fresh letter~$\$$, whose transitions are
represented in the trie as well.  We organize them such that in
phase~$\ell$ the $\$$-transitions for the states $q$~and~$p$ are the
same if and only if $\max{} \makeset{d(\delta(q, a), \delta(p, a))} {a
  \notin \sig(q)} \leq \ell - 1$.  This is easily organised if the
distance forest for all states with a finite right-language is
supplied.

\begin{theorem}
  \label{thm:kminimise partial}
  Given a (non-total) DFA~$M$ we can build a distance forest for it
  using $\mathcal{O}(\abs M \log n)$~linear-dictionary operations.
\end{theorem}


\section{Hyper-equivalence and hyper-minimisation}
\label{sec:hyper}

When considering minimisation with errors, it is natural that one
would like to impose a bound on the total number of errors introduced
by minimisation.  In this section, we investigate whether given $m, s
\in \nat$ and a DFA~$M$ we can construct a DFA~$N$ such that:
\begin{enumerate}
\item $N$~is hyper-equivalent to~$M$; i.e., $N \sim M$,
\item $N$ has at most $s$~states, and
\item $N$ commits at most $m$~errors compared to~$M$; i.e., $\abs{L(N)
    \simdiff L(M)} \leq m$. 
\end{enumerate}
Let us call the general problem `error-bounded hyper-minimisation'.
We show that this problem is intractable (NP-hard).  Only having a
bound on the number of errors allows us to return the original DFA,
which commits no errors.

To show NP-hardness of the problem we reduce the 
3-colouring problem 
to it.  Roughly speaking, we construct the DFA~$M$ from a graph~$G =
\langle V, E\rangle$ as follows.  Each vertex~$v \in V$ is represented
by a state~$v \in Q$, and each edge~$e \in E$ is represented by a
symbol~$e \in \Sigma$. We introduce additional states in a way such
that their isomorphic copies are present in any minimal DFA that is
hyper-equivalent to~$M$.  The additional states are needed to ensure
that for every edge $e = \{v_1, v_2\} \in E$ the languages
$L_M(\delta(v_1,e))$~and~$L_M(\delta(v_2, e))$ differ.  Now we assume
that $m = \abs E \cdot (\abs V - 2)$ and $s = 14$.  We construct the
DFA~$M$ such that all vertices of $V \subseteq Q$ are hyper-equivalent
to each other and none is hyper-equivalent to any other state.  We can
save $\abs V - 3$~states by merging all states of~$V$ into at most
$3$~states.  These merges will cause at least $\abs E \cdot (\abs V -
2)$~errors.  Additionally, $3$~states will become superfluous after
the merges, so that we can save $\abs V$~states.  There are two cases:
\begin{itemize}
\item If the input graph~$G$ is $3$-colourable by~$c \colon V \to
  [3]$, then we can merge all states of~$c^{-1}(i)$ into a~single
  state for every $i \in [3]$.  Since $c$~is proper, we never merge
  states $v_1, v_2 \in Q$ with $\{v_1, v_2\} \in E$, which avoids
  further errors.
\item On the other hand, if $G$~is not $3$-colourable, then we merge
  at least two states~$v_1, v_2 \in Q$ such that $e = \{v_1, v_2\} \in
  E$.  This merge additionally introduces $2$~errors caused by the
  difference $L(\delta(v_1, e)) \simdiff L(\delta(v_2, e))$.
\end{itemize}
Consequently, a DFA that (i)~is hyper-equivalent to~$M$, (ii)~has at
most $s$~states, and (iii)~commits at most $m$~errors exists if and
only if $G$~is $3$-colourable.  This shows that error-bounded
hyper-minimisation is NP-hard.

\begin{figure}
  \centering
  \includegraphics[width=\textwidth]{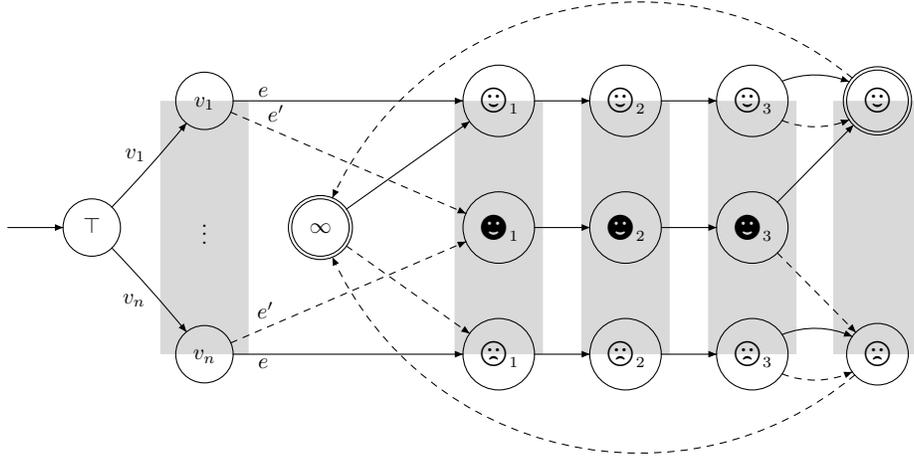}
  \caption{DFA~$M$ constructed in~\protect{Section~\ref{sec:hyper}},
    where $a$-transitions are represented by unbroken lines (unless
    noted otherwise), $b$-transitions by dashed lines,
    and $e = \{v_1, v_n\}$ and $e' = \{v_2, v_3\}$ with $v_1 < v_2 <
    v_3 < v_n$.  The hyper-equivalence~$\sim$ is indicated.}
  \label{fig:Hyper}
\end{figure}

\label{sec:HCon}

\begin{definition}
  \label{df:HInp}
  We construct a DFA~$M = \langle Q, \Sigma, \delta, \top, F \rangle$
  as follows:
  \begin{itemize}
  \item $Q = \{\top, \bot, \infty, \smiley, \frowny\} \cup V \cup
    \{\Circle_j \mid \Circle \in \{\smiley, \blacksmiley, \frowny \},
    j \in [3]\}$,
  \item $\Sigma = \{a, b \} \cup V \cup E$,
  \item $F = \{\infty, \smiley\}$,
  \item for every $v \in V$, $e = \{v_1, v_2\} \in E$ with $v
    \notin e$ and $v_1 < v_2$, $\Circle \in \{\smiley, \frowny \}$
    \begin{align*}
      \delta(\top, v) &= v & \delta(\infty, a) &= \smiley_1 &
      \delta(\infty, b) &= \frowny_1 \\
      \delta(v, e) &= \blacksmiley_1 & \delta(v_1, e) &= \smiley_1 &
      \delta(v_2, e) &= \frowny_1 \\[1ex]
      \delta(\blacksmiley_1, a) &= \blacksmiley_2 &
      \delta(\blacksmiley_2, a) &= \blacksmiley_3 &
      \delta(\blacksmiley_3, a) &= \smiley & \delta(\blacksmiley_3, b)
      &= \frowny & \\
      \delta(\Circle_1, a) &= \Circle_2 & \delta(\Circle_2, a) &=
      \Circle_3 & \delta(\Circle_3, a) &= \Circle & \delta(\Circle_3,
      b) &= \Circle & \delta(\Circle, b) &= \infty 
    \end{align*}
  \item For all remaining cases, we set $\delta(q, \sigma) = \bot$.
  \end{itemize}
\end{definition}

Consequently, the DFA~$M$ has $14 + \abs V$~states (see
Figure~\ref{fig:Hyper}).  Next, we show how to collapse
hyper-equivalent states using a proper $3$-colouring~$c \colon V \to
[3]$ to obtain only 14~states.

\begin{definition}
  \label{df:HOut}
  Let $c \colon V \to [3]$ be a proper $3$-colouring for~$G$.  We
  construct the DFA $c(M) = \langle P, \Sigma, \mu, \top, F \rangle$
  where
  \begin{itemize}
  \item $P = \{\top, \bot, \infty, \smiley, \frowny\} \cup [3]
    \cup \{\Circle_j \mid \Circle \in \{\smiley, \frowny\}, j \in
    [3]\}$,
  \item $\mu(p, \sigma) = \delta(p, \sigma)$ for all $p \in P
    \setminus \{\top, 1, 2, 3\}$ and $\sigma \in \Sigma$, and
  \item for every $v \in V$, $i \in [3]$, and $e = \{v_1, v_2\} \in E$
    with $v_1 < v_2$
    \begin{align*}
      \mu(\top, v) &= c(v) & \mu(i, e) &= \begin{cases}
        \smiley_1 & \text{, if } c(v_2) \neq i \\
        \frowny_1 & \text{, otherwise.}
      \end{cases}
    \end{align*}
  \end{itemize}
\end{definition}

\label{sec:HProof}
\begin{lemma}
  \label{lm:hmain}
  There exists a DFA that has at most $14$~states and commits at most
  $\abs E \cdot (\abs V - 2)$ errors when compared to~$M$ if and only
  if $G$~is $3$-colourable.
\end{lemma}

\begin{corollary}
  \label{cor:hyper}
  `Error-bounded hyper-minimisation' is NP-complete. More formally,
  given a DFA~$M$ and two integers $m, s \in \mathrm{poly}(\abs M)$,
  it is NP-complete to decide whether there is a DFA~$N$ with at most
  $s$~states and $\abs{L(M) \simdiff L(N)} \leq m$.
\end{corollary}


\section{Error-bounded $k$-minimisation}
\label{sec:kMinHardness}
In Section~\ref{k-minimisation} the number of errors between $M$~and
the constructed \mbox{$k$-minimal} DFA was not calculated.  In
general, there is no unique $k$-minimal DFA for~$M$ and the various
$k$-minimal DFAs for~$M$ can differ in the number of errors that they
commit relative to~$M$.  Since several dependent merges are performed
in the course of \mbox{$k$-minimisation}, the number of errors between
the original DFA~$M$ and the resulting \mbox{$k$-minimal} DFA is not
necessarily the sum of the errors introduced for each merging step.
This is due to the fact that errors made in one merge might be
cancelled out in a subsequent merge.
It is natural to ask, whether it is nevertheless possible to
\emph{efficiently} construct an \emph{optimal $k$-minimal} DFA for~$M$
(i.e., a~$k$-minimal DFA with the least number of errors introduced).
In the following we show that the construction of an optimal
$k$-minimal DFA for~$M$ is intractable (NP-hard).


\label{sec:kOverview}
The intractability is shown by a reduction from the $3$-colouring
problem for a graph $G = \langle V, E\rangle$ in a~similar, though
much more refined, way as in Section~\ref{sec:hyper}.  We again
construct a~DFA~$M$ with one state~$v$ for every vertex~$v \in V$ and
one letter~$e$ for each edge~$e \in E$.  We introduce three additional
states $\{1_0, 2_0, 3_0\}$ (besides others) to represent the
$3$~colours.  For the following discussion, let $N = \langle P,
\Sigma, \mu, p_0, F' \rangle$~be a $k$-minimal DFA for~$M$.  Let us
fix an edge $e = \{v_1, v_2\} \in E$.  The DFA~$M$ is constructed such
that the languages $L_M(\delta(v_1, e))$~and~$L_M(\delta(v_2, e))$
have a large but finite symmetric difference; as in the previous
section, if a proper $3$-colouring $c \colon V \to [3]$ exists the
DFA~$N$ can be obtained by merging each state $v$ into $c(v)_0$.  In
addition, for every edge $e = \{v_1, v_2\} \in E$ and vertex~$v \in
e$, we let~$\mu(c(v)_0, e) = \delta(v, e)$.  On the other hand, if
$G$~admits no proper $3$-colouring, then the DFA~$N$ is still obtained
by state merges performed on~$M$.  However, because $G$~has no proper
$3$-colouring, in the constructed DFA $M$ there exist $2$ states $v_1$,
$v_2$ such that $e = \{v_1, v_2\} \in E$ and that both $v_1$ and $v_2$
are merged into the same state~$p \in P$.  Then the transition $\mu(p,
e)$ cannot match both $\delta(v_1, e)$~and~$\delta(v_2, e)$.  In order
to make such an error costly, the left languages of $v$~and~$v'$ are
designed to be large, but finite.  In contrast, we can easily change
the transitions of states~$\{1_0, 2_0, 3_0\}$ by letters~$e$ because
the left-languages of the states~$\{1_0, 2_0, 3_0\}$ are small.


To keep the presentation simple, we will use two gadgets.  The first
one will enable us to make sure that two states cannot be merged:
$k$-similar states are also hyper-equivalent, so we can simply avoid
undesired merges by making states hyper-inequivalent.  Another gadget
will be used to increase the in-level of certain states to a desired
value.


\begin{lemma}
  \label{lm:kCong}
  For every congruence~$\mathord{\simeq} \subseteq Q \times Q$ on~$M$,
  there exists a DFA~$N$ 
  such that
  (i)~$p_1 \not\sim p_2$ for every $p_1 \in P \setminus Q$ and $p_2
  \in P$ with $p_1 \neq p_2$, and
  (ii)~$q_1 \not\sim q_2$ in~$N$ for all $q_1 \not\simeq q_2$.
\end{lemma}

In graphical illustrations, we use different shapes for
$q_1$~and~$q_2$ to indicate that $q_1 \not\sim q_2$, because of the
gadget of Lemma~\ref{lm:kCong}.  Note that 
states with the same shape need not be $k$-similar.

\begin{lemma}
  \label{lm:kInlevel}
  For every subset~$S \subseteq Q \setminus \{q_0\}$ of states and
  map~$\mathord{\minlevel} \colon S \to \nat$, there exists a DFA~$N =
  \langle Q \cup I, \Sigma \cup \Delta, \mu, q_0, F \rangle$ such that
$\abs{\mu^{-1}(i)} = 1$ for every $i \in I$ and 
$\inlevel_N(s) \geq \minlevel(s)$ for every $s \in S$.
\end{lemma}

We will indicate the level~$i$ below the state name in
graphical illustrations.  Moreover, we add a special feathered arrow
to the state~$q$, whenever the gadget is used for the state~$q$ to
increase its level.

\label{sec:kCons}
Next, let us present the formal construction.  Let $G = \langle V,
E\rangle$ be an undirected graph.  Select $k, s \in \nat$ such that $s
> \log(\abs V) + 2$ and $k > 4s$.  Moreover, let $\ell = k - 2s$.
\vspace{-0.7cm}
\begin{figure}
  \centering
  \includegraphics[scale=0.61]{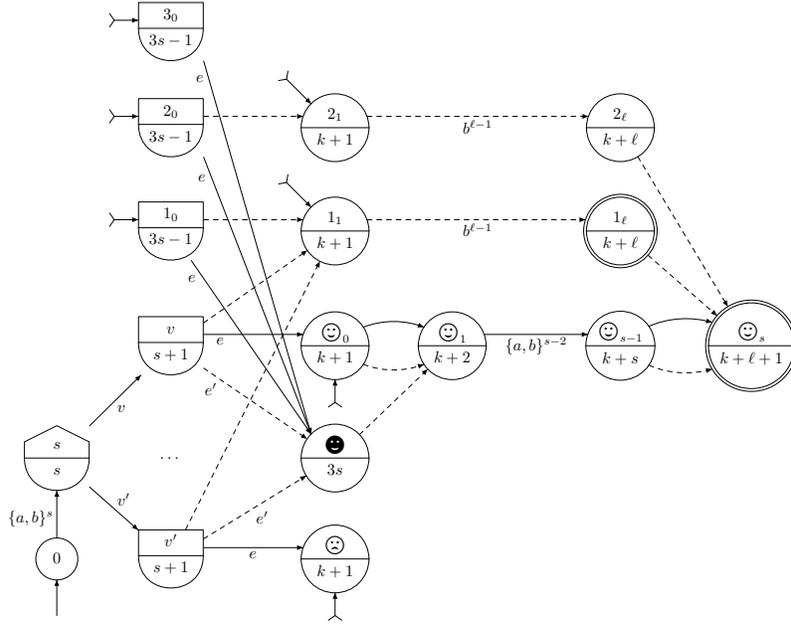}
  \caption{Illustration of the DFA~$M$ of
    Section~\protect{\ref{sec:kMinHardness}}}
  \label{fig:main_states}
\end{figure}
\vspace{-0.7cm}
\begin{definition}
  \label{df:kInp}
  We construct the DFA~$M = \langle Q, \Sigma, \delta, 0, F \rangle$
  as follows:
  \begin{itemize}
  \item $Q = \{\bot, \blacksmiley, \frowny, 3_0\} \cup \{i_j \mid i
    \in [2], j \in [\ell]\} \cup V \cup [0, s] \cup \{ \smiley_i \mid
    0 \leq i \leq s\}$,
  \item $\Sigma = \{a, b\} \cup V \cup E$,
  \item $F = \{\smiley_s, 1_{\ell}\}$, and
  \item for every $v \in V$, $e = \{v_1, v_2\} \in E$ with $v \notin
    e$ and $v_1 < v_2$, $i \in [s]$, and $j \in [\ell]$
    \begin{align*}
      \delta(i-1, a) &= i & \delta(v_1, e) &= \smiley_0 & \delta(1_0,
      e) &= \blacksmiley & \delta(1_{j-1}, b) &= 1_j \\
      \delta(i-1, b) &= i & \delta(v_2, e) &= \frowny & \delta(2_0, e) &=
      \blacksmiley & \delta(2_{j-1}, b) &= 2_j \\
      \delta(\smiley_{i-1}, a) &= \smiley_i & \delta(v, e) &=
      \blacksmiley & \delta(3_0, e) &= \blacksmiley & \delta(1_\ell,
      b) &= \smiley_s \\ 
      \delta(\smiley_{i-1}, b) &= \smiley_i & \delta(v, a) &= 1_1 & &&
      \delta(2_\ell, b) &= \smiley_s \\
      \delta(\blacksmiley, a) &= \smiley_1 & \delta(s, v) &= v
    \end{align*}
  \item For all remaining cases, we set $\delta(q, \sigma) = \bot$.
  \end{itemize}
\end{definition}

Finally, we show how to collapse $k$-similar states using a proper
\mbox{$3$-colouring} $c \colon V \to [3]$.  We obtain the $k$-similar
DFA~$c(M) = \langle P, \Sigma, \mu, 0, F \rangle$ from~$M$ by merging
each state~$v$ into~$c(v)_0$.  In addition, for every edge $e = \{v_1,
v_2\} \in E$, we let $\mu(c(v_1)_0, e) = \delta(v_1, e)$ and
$\mu(c(v_2)_0, e) = \delta(v_2, e)$.  Since the colouring~$c$ is
proper, we have that $c(v_1) \neq c(v_2)$, which yields that $\mu$~is
well-defined.  For the remaining $i \in [3] \setminus \{c(v_1),
c(v_2)\}$, we let $\mu(i_0, e) = \smiley_0$.
All equivalent states (i.e., $\bot$~and~$\frowny$) are merged.  The
gadgets that were added to~$M$ survive and are added to~$c(M)$.
Naturally, if a~certain state does no longer exist, then all
transitions leading to or originating from it are deleted too.  This
applies for example to~$\blacksmiley$.

\label{sec:kProof}
\begin{lemma}
  \label{lm:kMain}
  There exists a $k$-minimal DFA~$N$ for~$M$ with at most
  \[ 2^{2s-1} \cdot \abs E \cdot (\abs V - 2) + 3 \cdot 2^{s-1} \cdot
  \abs E + 2^{s+1} \cdot \abs V \]
  errors if and only if the input graph~$G$ is $3$-colourable.
\end{lemma}

\begin{corollary}
  \label{cor:khyper}
  `Error-bounded $k$-minimisation' is NP-complete.
\end{corollary}

%
%

\bibliographystyle{splncs03}
\bibliography{biblio}

\appendix
\clearpage

\section{Proofs and additional material for Section~\ref{sec:prel}}

\begin{lemma}
\label{lemat o nieskonczonosci}
If $p,q \in Q$ then $d(p,q)<+\infty$ implies that $d(p,q)<n$.
\end{lemma}
\begin{proof}
  Let $D_i$ denote the equivalence relation defined as $D_i(p,q)$ iff
  $d(p,q) \leq i$.  Let $n_i$ be the number of equivalence classes of
  $D_i$, for $i = 0, 1 , \dots$.  Note that if $n_i = n_{i+1}$ then
  $n_i = n_{j}$ for all $j > i$ and $d(q, q') > i$ implies $d(q,q') =
  +\infty$.

  Since $n_0 \leq \abs Q$ the sequence $n_0 \geq n_1 \geq \ldots$
  stabilises at position $n_{|Q|-1}$, i.e., there are no states $q,q'$
  such that $D_{|Q|}(q,q')$ and $\neg D_{|Q|-1}(q,q')$.  Hence
  $d(q,q')< + \infty$ implies $D_{|Q|-1}(q,q')$, i.e., $d(q,q')<n$.
  \qed
\end{proof}

\section{Proofs and additional material for
  Section~\ref{k-minimisation}}
\subsection{Proofs and additional material for
  Section~\ref{k-minimisation2}}

It can be shown that if $M \sim_k N$ then the states reached after
reading the same word are also $k$-similar, assuming that the word is
short enough.

\begin{lemma}
  \label{lm:notsim}
  Let $M \sim_k N$, $q_1, q_2 \in Q$, and $w_1, w_2 \in \Sigma^*$ be
  such that $\delta(w_i) = q_i$ and $\abs{w_i} = \inlevel_M(q_i)$ for
  $i \in [2]$.  If $q_1 \not\sim_k q_2$ , then $\mu(w_1) \not\sim_k
  \mu(w_2)$.
\end{lemma}
\begin{proof}
  First, suppose that $q_1 \not \sim q_2$. Then, $M \sim N$ yields
  that
  \[ \mu(w_1) \sim \delta(w_1) = q_1 \not\sim q_2 = \delta(w_2) \sim
  \mu(w_2) \] and thus $\mu(w_1) \not\sim \mu(w_2)$, which proves that
  $\mu(w_1) \not\sim_k \mu(w_2)$.
  
  Second, let $d(q_1, q_2) < \infty$.  Since $q_1 \not\sim_k q_2$, we
  have 
  \[ d(q_1, q_2) + \min(k, \abs{w_1}, \abs{w_2}) > k \enspace. \]
  Clearly, there exists $u \in L_M(q_1) \simdiff L_M(q_2)$ with $\abs
  u \geq d(q_1, q_2)-1$.  Moreover, $\abs{w_1u} \geq k \leq
  \abs{w_2u}$.  Since $M \sim_k N$, we have $w_1u, w_2u \notin L(M)
  \simdiff L(N)$.  Consequently
  \[
  u \notin L_M(q_1) \simdiff L_N(\mu(w_1)) \qquad \text{and} \qquad u
  \notin L_M(q_2) \simdiff L_N(\mu(w_2)) \enspace. \] By assumption,
  $u \in L_M(q_1) \simdiff L_M(q_2)$ and thus $u \in
  L_N(\mu(w_1)) \simdiff L_N(\mu(w_2))$, which shows that $d(\mu(w_1),
  \mu(w_2)) \geq \abs u + 1$.  Clearly, $\inlevel_N(\mu(w_1)) \geq
  \abs{w_1}$ and $\inlevel_N(\mu(w_2)) \geq \abs{w_2}$, which yields
  $\mu(w_1) \not\sim_k \mu(w_2)$.  \qed
\end{proof}

We show some properties of \kminimise{}, which are used to show that
it properly constructs a $k$-minimal DFA.  Let $N$ denote the DFA
constructed by \kminimise{} at any particular point.

\begin{lemma}
  \label{lem:leading word}
  If $\delta_N(p', w) = p$ and $\inlevel_M(p) < k$ then
  \begin{equation}
    \label{eq:leading word}
    |w| \leq \inlevel_M(p) - \inlevel_M(p').
  \end{equation}
\end{lemma}

\begin{proof}
  The assertion of the lemma is shown to be hold after each merge done
  by \kminimise{}, i.e., by the induction on the number of merges done
  by \kminimise.  If there were no merges done yet then $N=M$ and the
  claim holds true.  Let $N$ denote the DFA before the merge and $N'$
  after it.

  We focus on $w=a \in \Sigma$.  So assume that
  $\delta_{N'}(p_1,a)=p_2$ after merging state $p$ to $q$.  The only
  non-trivial case is when $\delta_{N}(p_1,a)=p$ and $p_2=q$, i.e.,
  when something is changed after the merging.  By induction
  assumption $\inlevel_M(p)-\inlevel_M(p_1)\geq 1$.  As
  $\inlevel_M(q)\geq \inlevel_M(p)$ as guaranteed by \kminimise, the
  claim is obtained.

  When $|w|>1$ it is enough to consider the states obtained after
  transitions after each letter of $w$ and sum up the inequalities.
  \qed
\end{proof}

\begin{lemma}
\label{lem:distance is small}
During the run of \kminimise{} for all $p'\in Q(N)$, 
\begin{equation}
\label{eq:distance_preserved}
d(L_{M}(p'),L_{N}(p'))
	\leq
\max(0,k - \inlevel_M(p')).
\end{equation}
\end{lemma}
\begin{proof}
  We establish this claim by induction.  Let $N'$ denote the DFA after
  merging $p$ to $q$ and $N$ just before this merge.  Note, that as $p
  \not \equiv q$ (in $M$), thus $d(p,q)>0$ (in $M$).  Thus $p \sim_k
  q$ implies $\min(k,\inlevel_M(p),\inlevel_M(q))<k$.  Since $p$ is
  merged to $q$ by \kminimise, $\inlevel_M(p) \leq \inlevel_M(q)$ and
  as $p \sim_k $ q also $\inlevel_M(p)<k$.  Then by
  Lemma~\ref{eq:leading word} we conclude that there is no word
  leading from $q$ to $p$ in $N$: assume for the sake of contradiction
  that there is such a word $w$.  Since $\inlevel_M(p) < k$, by
  Lemma~\ref{eq:leading word}
  \[ \inlevel_M(p) \geq \inlevel_M(q) + |w| > \inlevel_M(q) \enspace ,
  \]
  contradiction.  Thus there is now word leading from $q$ to $p$ in
  $N$ and therefore $L_N(q) = L_{N'}(q)$.

  For the other case, lest us first estimate $d(L_{N}(p),L_{N'}(q)) =
  d(L_{N}(p),L_{N}(q))$.  As already noted,
  $\min(k,\inlevel_M(p),\inlevel_M(q)) = \inlevel_M(p)$, which allows
  us to reduce $p \sim_k q$ to
  \begin{align}
    \notag d(L_M(p),L_M(q)) + \inlevel_M(q) &\leq k \intertext{and
      thus} \notag d(L_M(p),L_M(q)) & \leq k - \inlevel_M(p).
    \intertext{By induction assumption} \notag d(L_M(q),L_{N}(q))
    &\leq
    \max(0,k - \inlevel_M(q))\\
    \notag d(L_M(p),L_{N}(p)) &\leq \notag
    \max(0,k - \inlevel_M(q))\\
    &= \notag k - \inlevel_M(p) \intertext{and as $d$ is an ultra
      metric}
    \label{eq:distance_induction_works}
    d(L_N(p),L_{N}(q)) &\leq k - \inlevel_M(p) \enspace .
  \end{align}

  So consider an arbitrary state $p'$.  If it has no word leading to
  $p$ in $N$, then $L_N(p')=L_{N'}(p')$ and we are done.  If it has a
  word $w$ leading to $p$, then Lemma~\ref{lem:leading word} can be
  applied, establishing:
  \begin{align*}
    d(L_N(p'),L_{N'}(p')) &= \max_{w: \delta_N(p',w)=p}
    |w|+d(L_{N'}(p),L_{N}(q)) \intertext{the former can be estimated
      by~\eqref{eq:distance_preserved} and the latter
      by~\eqref{eq:distance_induction_works}, yielding}
    d(L_N(p'),L_{N'}(p')) &\leq
    (\inlevel_M(p)-\inlevel_M(p'))+(k-\inlevel_M(p))\\
    &\leq k - \inlevel_M(p'),
  \end{align*}
  which ends the proof.  \qed
\end{proof}

\begin{proof}[of Theorem~\ref{thm:kminimise is proper}]
Let $q_0,q_1,\ldots,q_n$ be the starting states in DFAs $M=N_0$, $N_1$,
\ldots , $N_n=N$.  By Lemma~\ref{lem:distance is small},
$d(L_M(q_i),L_{N_i}(q_i))\leq k$.
On the other hand,
since $q_i$ is merged to $q_{i+1}$ then $d(L_M(q_i),L_M(q_{i+1}))\leq k$.
So all the languages in question are within distance $k$ of each other and therefore
$$
d(L(M),L(N))\leq k \enspace .
$$
Thus $M \krowne N$. It is left to show that $N$ is $k$-minimal.
Consider the set of states $Q'$ of $N$ and let $M'$ be a DFA $k$-similar to $M$.
By \kminimise{}, they are pairwise
$k$-dissimilar (as states in $M$). For a state $q \in Q'$ let $w_q$
be the word such that $|w_q| \geq \min(\inlevel_M(q),k)$.
Consider any two such words $w_q$ and $w_p$.
Then by Lemma~\ref{lm:notsim} $w_q$ and $w_p$
cannot lead to the same state in $M'$.
Hence the size of $M'$ is at least $|Q'|$, which is exactly the size of $N$.
\qed
\end{proof}

\begin{corollary}[of Theorem~\ref{thm:kminimise is proper}]
\label{cor:kSet}
Each maximal (with respect to the inclusion) set $Q'$
of pairwise $k$-dissimilar states of a DFA $M$ is of size
of the $k$-minimal DFA for $M$.
\end{corollary}
\begin{proof}
First note that without loss of generality we may assume that
\begin{equation}
\label{eq:maximal elements}
p \in Q', \ q \notin Q' \text{ and } p \krowne q \text{ implies } \inlevel_M(p) \geq \inlevel_M(q).
\end{equation}
If not, then we can replace $p$ by $q$ in $Q'$, without loosing the assumed property of $Q'$.
After finitely many such substitutions, $Q'$ satisfying~\eqref{eq:maximal elements} is obtained.

Run \kminimise{} for $M$, and whenever there are two states $p \krowne q$ considered,
merge the one outside $Q'$ to the one in $Q'$ (do arbitrarily, if none is in $Q'$).
Since $Q'$ is maximal with respect to the inclusion, \kminimise{} terminates with the DFA with
$Q'$ as the set of states.
\qed
\end{proof}

Now we are able to establish a structural characterisation of $k$-similar DFAs,
analogous to characterisation of hyper-equivalent DFA's~\cite[Sect.~3.2]{BadrRAIRO}.
In particular, we derive the analogue of~\cite[Theorem~3.8]{BadrRAIRO} for $k$-similar DFAs.

\begin{corollary}[{\protect{of Lemma~\ref{lm:notsim} and Corollary~\ref{cor:kSet}}}]
  \label{cor:multisets}
  Let $S \subseteq Q$ be a maximal set of pairwise $k$-dissimilar
  states of~$M$, and let $N$~be a~$k$-minimal DFA for~$M$.  Then there
  exists a bijection $h \colon S \to P$ such that
  \begin{itemize}
  \item $q \sim h(q)$ for every $q \in S$, and
  \item $q \equiv h(q)$ for every $q \in S$ such that $\inlevel_M(q)
    \geq k$. 
  \end{itemize}
\end{corollary}
\begin{proof}[of Corollary~\ref{cor:multisets}]
  We have $\abs S = \abs P$ by Corollary~\ref{cor:kSet}.  For every $q \in
  S$, let $w_q \in \delta^{-1}(q)$ be such that $\abs{w_q} =
  \inlevel_M(q)$.  We define the mapping $h \colon S \to P$ by $h(q) =
  \mu(w_q)$ for every $q \in S$.  Since $M \sim_k N$, which yields $M
  \sim N$, we have $q = \delta(w_q) \sim \mu(w_q) = h(q)$.  Finally,
  suppose that $\inlevel_M(q) \geq k$.  Then $L_M(q) =
  L_M(\delta(w_q)) = L_N(\mu(w_q)) = L_N(h(q))$ because $M \sim_k N$, 
  which yields $q \equiv h(q)$.  \qed
\end{proof}

\subsection{Additional material for Section~\ref{subsec:distance forests}}
We refer to the distance tree we construct for the DFA $M$ using the
notation $\mathcal{D}(M)$.  We identify the leaves with the states of
the DFA if this raises no confusion.  To simplify the argument, we
assume that $\bot$ is always in the $\mathcal{D}(M)$.  We refer to a
tree in a distance forest by a name of a \emph{distance tree}.  The
vertices that are present in the compressed representation are called
\emph{explicit}, while those that were removed are called
\emph{implicit}.  The standard terms \emph{father} $\father(v)$ of a
vertex $v$ and \emph{ancestor} always refer to implicit vertices.

\begin{algorithm}
\caption{\distancetree{}\label{alg:simple distance forest}}
\begin{algorithmic}[1]
\For{$p \in Q$}
	\State $\state(p) \gets p$, $\level(p) \gets 0$, activate $p$
\EndFor
\For{$\ell=1$ \To{} $|n|$}
\label{line:grouping}
\State group active nodes according to $(\delta((\state(v),a))_{a \in \Sigma})$
\For{each group of nodes $V'$ such that $|V'|>1$}
	\State choose $v \in V'$
	\State create active node $v'$, $\level(v')=\ell$, $\state(v')=\state(v)$
	\For{$v'' \in V'$}
		\State join $v''$ to $v'$, deactivate $v''$ 
		\State replace $\state(v'')$ in entries of $\delta$ by $\state(v')$
	\EndFor
\EndFor
\EndFor
\end{algorithmic}
\end{algorithm}

\begin{lemma}
For a minimised DFA $M$ if \distancetree{}
replaced the state $q$ in $\delta$ by $p$ at phase $\ell$ then $d(p,q)=\ell$.
\end{lemma}
\begin{proof}
The proof proceeds by induction on $\ell$.
If $\ell=1$ then $q$ and $p$ have the same successors. Since $q$ and $p$ are not equivalent,
their distance is exactly $1$, as claimed.

Suppose that $q$ was replaced by $p$ in phase $\ell>1$.
Since $q$ was not replaced by $p$ in phase $\ell-1$, by the induction assumption $d(q,p) > d-1$.
Let $(q_1, \ldots , q_{|\Sigma|})$ be the set of successors of $q$ in the DFA
and let $(q_1', \ldots , q'_{|\Sigma|})$ be the vector of its successors in the representation in \distancetree.
Since $q_i$ was replaced by $q_i'$ in phase $\ell-1$ or earlier, $d(q_i,q_i') \leq \ell-1$.
Similarly, $d(p_i,q_i') \leq \ell-1$.
Then, by~\eqref{eq:distance recursion},
\begin{align*}
d(p,q)
	&=
1 + \max_{i=1}^{|\Sigma|} d(q_i,p_i)\\
	&\leq
1 + \max_{i=1}^{|\Sigma|} \left(\max(d(q_i,q_i'),d(q_i',p_i))\right)\\
	&\leq
\ell.
\end{align*}
Consider now any other state $r$ such that $r$ was replaced by $q$ in earlier phases.
Then $d(r,q)\leq \ell-1$ and thus $d(r,p) = \ell$.
\qed
\end{proof}

The bottleneck of \distancetree{}
is the replacing of occurrences of $q$ in $\delta$ by some other state $p$.
We show that such replacing can be done in a way so that a single entry in $\delta$
is modified at most $\log n$ times.

\begin{lemma}
\label{lem:delta entry replaced log n times}
\distancetree{} can be implemented so that it alters every value of $\delta$ at most $\log n$ times.
\end{lemma}
\begin{proof}
The key modification needed in \distancetree{}
is the choice of node $v \in V'$. For each state we introduce a counter
$c(p)$, initially set to $1$, which keeps the track of how many states $p$ represents.
When we choose a node $v \in V'$ we take the one with the largest $c(\state(v))$.
We update the value accordingly $c(\state(v')) \gets \sum_{v'' \in V'} c(\state(v''))$.

Note, that if we replace a value $q$ by $p$ in $\delta$,
then $c(p)\geq c(q)$ before the update of $c$ and so $c(p)\geq 2 c(q)$ after the change.
Thus if we replace the entry in $\delta$, the corresponding value of $c$ at least doubles.
Since $c(p)$ is upper-bounded by $n$, each entry is replaced at most $\log n$ times.
\qed
\end{proof}

\begin{proof}[of Theorem~\ref{thm:linear dictionary log n times}]
To allow fast replacing of entries $q$ in the $\delta$,
for each state $q$ used as the label in one of the dictionaries,
we store an up-to-date list of its occurrences in all vectors in all dictionaries.

When state $p$ is merged into $q$ we update the tries:
we use the up-to date list of occurrences.
For each occurrence of $p$ in an internal node $v$ we have the following situations:
if $v$ does not have a child labelled by $q$ then
we remove $p$ from the linear dictionary, and insert $q$ into it, pointing
at the same child as $p$ used to.
If $v$ has both children $p$ and $q$,
we have to merge their corresponding subtrees, rooted at $v_1$ and $v_2$.
We choose one of them, say $v_1$, and insert each child of $v_1$ into subtrie of $v_2$.
Then we set pointer from $q$ to $v_2$.
This might result in yet another situation of the same type,
we do so recursively until we get to the leaves.

The total cost of the case, when we did not need to merge linear dictionaries
can be bounded similarly as in Lemma~\ref{lem:delta entry replaced log n times}:
note, that after inserting $q$ and deleting $p$ from the linear dictionary,
$c(q) \geq 2c(p)$. Thus each such element is modified at most $\log n$ times
and so the total cost is $\mathcal{O}(|\Sigma|n\log n)$.

If we do merge the linear dictionaries, we make a different analysis.
For each linear dictionary we keep a counter, which calculates how many vertices were
inserted into this linear dictionary.
When we merge two linear dictionaries, we remove the one with the smaller value of the counter and insert all its elements into the other dictionary.
Then we sum the counters and update the remanding counter.

Since each time a vertex is reinserted, the value of the counter in its linear dictionary at least doubles,
and the maximal value of such counter is $n$,
each vertex is inserted into a dictionary at most $\log n$ times.

If we are to merge two leaves, we simply join their respective lists and remove one of the leaves.
\qed
\end{proof}

Now we can use the distance forest to our benefit.  Instead of finding
pairs of states that are $k$-similar we proceed in another fashion:
roughly speaking, for each state $q$ we want to find the closest (with
respect to $d$) state $p$ satisfying $\inlevel(p) \geq \inlevel(q)$,
then, using this state, we want to judge, whether $q$ is ever going to
be merged to other state.  To this end, we label the nodes of the
$\mathcal D (M)$ by states of the DFA $M$, formally we define
$\state(v)$ for each node $v$ of $\mathcal D (M)$.  Since leaves of
the $\mathcal D (M)$ are identified with leaves, we obviously set
$\state(q) = q$ for each such leaf.  Then we label each inner node
with one of the labels of its children, choosing the one with the
maximal $\inlevel_M$.

For each state $q$ let its \emph{submit node} be the first node on the path
from leaf $q$ to the root labelled with a state different than $q$ and
let the \emph{submit state} be the label of this node;
let $d(q)$ be the depth of the submit node $q$,
if $q$~is the label of the root, then $d(q) = \infty$.
Furthermore, define $\values(q) = \inlevel_M(q) + d(q)$.
The next lemma shows that $\values(q)$ can be used to approximate $\sim_k$.
\begin{lemma}
\label{lem:values to merges}
If $\values(q) \leq k$ then $q$ is $k$-similar to every state appearing as the label
on the path from its submit state to the root.

If $\values(p), \values(q) > k$ then $p \sim_k q$.
\end{lemma}
\begin{proof}
Suppose that $\values(q) \leq k$ and let $q'$ be any node label above $q$'s submit vertex (inclusively). Then
\begin{align*}
\min(\inlevel(q),\inlevel(q')) + \level(\lca(q,q'))
	&=
\inlevel(q) + \level(\lca(q,q'))\\
	&\leq
\inlevel(q) + d(q)\\
	&=
\values(q)\\
	&\leq
k \enspace,
\end{align*}
and so $q \sim_k q'$.

Let $\values(q),\values(q') > k$, without loss of generality we may
assume that $\inlevel(q) \leq \inlevel(q')$ and $q$ is not $q'$'s submit node
(which could happen if $\inlevel(q) = \inlevel(q')$.
Then $d(q,q') = \level(\lca(q,q')) \geq d(q)$. Thus
\begin{align*}
d(q,q') + \min(\inlevel(q),\inlevel(q'))\
	&\geq
d(q) + \inlevel(q) \\
	&=
\values(q)\\
	&>
k \enspace,
\end{align*}
which concludes the proof. 
\qed
\end{proof}

For each state $q$ let its $k$-ancestor node $v$ be the first node
on the path from $q$ to the root such that $\values(\state(v)) > k$,
and let $k$-ancestor state be the label of $v$.

\begin{corollary}[of Lemma~\ref{lem:values to merges}]
  \label{cor:merge-tree}
  The DFA obtained by merging each state $q$ to its $k$-ancestor state
  is $k$-minimal and $k$-similar to~$M$.
\end{corollary}

The following theorem states easy consequences of Corollary~\ref{cor:merge-tree}.

\begin{proof}[of Theorem~\ref{thm:k minimisation implementation}]
  Calculate the labels in the $\mathcal D(M)$ as described earlier,
  this can be done using one depth-first traversal. Then calculate
  $\values(q)$ for each state $q$, this as well can be done using one
  depth-first traversal.  Sort the pairs $(q,\values(q))$ according to
  $\values(q)$, since $\values(q) \leq 2n$, this can be done in linear
  time using \proc{CountingSort}.  Note, that the number of states of
  $k$-minimal DFA $M_k$ equals $\abs{\makeset{q}{\values(q) > k}}$, by
  Corollary~\ref{cor:merge-tree}, which can be now easily computed in
  linear time.

By Corollary~\ref{cor:merge-tree} to obtain the $k$-minimal DFA
it is enough to merge each $q$ with $\values(q)\leq k$ to its $k$-ancestor.
A table of $\func{ancestor}$ assigning to state $q$ its $k$-ancestor
can be computed in linear time.
The merging can be performed in $\mathcal{O}(|\delta|)$ time,
as we are only interested in the transition of the states $q'$ such that $\values(q') > k$.
We look through $\delta$ and replace each entry $\delta(q',a) = q$ by
$\func{ancestor}(q)$.

There is a little subtlety: when replacing $\delta(q',a) = q$ by $\delta(q',a) = q''$
we should take care that $q \neq \bot$, as otherwise it would be impossible to
bound the running time by $|\delta|$. However, note that if $q'' \sim_k \bot$ then
the language of $q''$ is finite and therefore there is a path from $q''$ to $\bot$,
hence $\inlevel(q'') \leq \inlevel(\bot)$. Note, that when $\inlevel$ for two states
are equal, we arbitrarily choose one, so without loss of generality it can be assumed
that $\bot$ is never merged to any other state.

We now present a $\mathcal{O}(|\delta| \log n)$ algorithm,
which at step $k=0, 1 , \ldots ,n$ has in memory the $k$-minimal DFA.
As previously, in step $k$ it will keep only states with $\values$ at greater than $k$:
it merges each existing state $q$ such that $\values(q) = k$ into its
$k$-ancestor $q'$ which by the construction satisfies $\values(q') \geq \values(q) + 1 > k$.
To obtain the proper running time, we need only need to organise the data structures properly.
It is represented by a list of transition:
for each state $q$ we list the pairs $(a,q')$ such that $\delta(a,q)=q'$, for all valid $a$.
Moreover, each $q$ has a list of incoming transition, i.e.,
list of pointers to the transitions to it.
Moreover, each state $q$ has a counter $\rank(q)$,
which describes how many states were merged to it.

Assume that we merge $p$ to $q$ and $\rank(p) \leq \rank(q)$.
Then the situation is easy. We redirect each transition to $p$ into
$q$, and perform the update
$\rank(q) \gets \rank(q) + \rank(p)$.
When $\rank(p) > \rank(q)$ then we redirect each transition to $q$ into $p$,
and replace the outgoing transition from $p$ by outgoing transitions from $q$.
Since they are given as a list, this is done in $\O (1)$ time.
Then we rename $p$ as $q$ and update rank $\rank(q) \gets \rank(q) + \rank(p)$. 

Note, that each time an entry in $\delta$ is modified, the $\rank$ of the target state doubles.
Thus each transition is modified at most $\log n$ times and so the running time is
$\mathcal{O}(|\delta| \log n)$.

Note, that while the running time $\mathcal{O}(|\delta| \log n)$,
outputting the results for each $k$ might take a time up to $\Omega(n |\delta| \log n)$.
\qed
\end{proof}

\subsection{Additional material for Section~\ref{subsec:partial transition function}}

\begin{lemma}
\label{lem:finite states cost}
The total cost of maintaining the transition by $\$$ in the tries is
$\mathcal{O}(|\delta|\log n)$ linear dictionary operations.
\end{lemma}
\begin{proof}
Because $L(M)$ is finite, $m(p) = \max\{|w| : w\in L(p)\}$ is defined for any state $p$,
let us denote $Q_i = \makeset{p}{m(p) = i}$ and $Q_{<\infty} = \bigcup_i Q_i$.

Consider a DFA $M'$ built on states $Q_{<\infty}$ and their direct predecessors,
i.e., $Q_{<\infty}' = \{q' \: : \: \exists a \in \Sigma \; \delta(q,a) \in Q_{<\infty}\} \cup \{ \bot \}$.
Take the $\delta'$ restricted to input from $Q_{<\infty}'$ and values in $Q_{<\infty}$.
For each state $q$ with infinite right-language and transitions into states in $Q_{<\infty}$
we insert into vectors of successors $(\$,q')$, or $(\$,\bot)$, if $q$ has only undefined transitions.
By Theorem~\ref{thm:finite distance} the cost of construction of $\mathcal T'$ for $M'$ is
$\mathcal{O}(|\delta|\log n)$.
Using $\mathcal T'$ for $M'$ in phase $d$ we merge states from $M'$ which are at distance $d-1$ or less:
it is enough to sort the nodes of $\mathcal T'$ according to their $\level$ and in phase
$d$ merge leaves in subtree of each node $v$ such that $\level(v) \leq d$.
To perform the merging efficiently,
each state in $M'$ is assigned $\rank$, which denotes the number of states that it represents.
When states $q_1, \ldots , q_r$ are to be merged, we choose the one with the maximal rank,
say $q_i$ and replace of occurrence of $q_1, \ldots , q_r$ in the trie by $q_i$.
Then we update the rank: $\rank(q_i) \gets \sum_{j=1}^r \rank(q_j)$.
Thus each such entry in the trie is replaced at most $\log n$ times: whenever it is replaced,
the corresponding $\rank$ doubles and it is upper bounded by $n$.
\qed
\end{proof}

\begin{proof}[of Theorem~\ref{thm:kminimise partial}]
In linear time we can identify states such that their right-language is finite.
By Theorem~\ref{thm:finite distance} their distance tree can be built in
$\mathcal{O}(|\delta|\log n)$ time.

Grouping of states in \distancetree{}
is a by-product of using trie for the vectors of successors.
Each state is deleted from set $Q$ once,
each such deletion results in an update of the trie.

The number of linear dictionary operations for the letters in the respective signature
can be upper bounded as in Theorem~\ref{thm:linear dictionary log n times},
that is, by $\O (|\delta|\log n)$.
By Lemma~\ref{lem:finite states cost} the same bound applies to the construction
and usage of the distance tree for states with finite right-language, 
\qed
\end{proof}

Next, we comment how the ancestors and nodes are represented for the
algorithm, as some of them are implicit: they are represented by an
explicit vertex directly below them with an offset, i.e., a pair
$(v,\ell)$.

Recall, that we iteratively construct fragments of the distance forest,
built on states $Q_t = \makeset{p}{m(p) = t}$, i.e., recognising words
of length at most $t$.
For each already constructed fragment we do the preprocessing allowing
efficient computing the $\lca$ for any pair of vertices.
There are known construction for doing this in constant time~\cite{FarachColtonBender}
however, they use the power of the full RAM model.
For our purposes, the simple construction that keeps at every node a list of ancestors
$2^0$, $2^1$, \ldots, $2^{\log n}$ higher allow performing the search in $\Theta(\log n)$ time,
which does not influence the total running time in our case.
As shown later, having the preprocessing performed for each such a fragment separately,
is enough to execute $\lca$-queries for the whole tree.

Firstly we argue that indeed adding fragments built on states from $Q_t$ is reasonable.
Moreover, if for some states $m(p)\neq m(q)$, $d(p,q)$ can be calculated easily.

\begin{lemma}
\label{lem:mp}
If $m(p) \neq m(q)$ then $d(p,q) = \max(m(p),m(q)) + 1$.\\
If $m(p) = m(q) $ then $d(p,q) \leq \max(m(p),m(q)) + 1$.
\end{lemma}

Recall, that the spine is the path joining the state $\bot$,
which recognises an empty language, with the root of the $\mathcal D(M)$.
Note, that the spine has no compressed fragments,
as the distance between $\bot$ and a a state $p$ is equal to $m(p)+1$,
by Lemma~\ref{lem:mp}, moreover, by~\eqref{eq:distance recursion}
it is easy to see that $\{m(p) \: : \: p \in Q_{<\infty} \}$
is equal to $\{ 0,1,2, \ldots , \max\{  p \in Q_{<\infty}\: : \: m(p)\}$.
Therefore the spine is created beforehand as an uncompressed line
of length $\max\{ p \in Q_{<\infty} \: : \: m(p)+1 \}$, which is at most $|Q_{<\infty}|$. 

\begin{lemma}
\label{lem:distancecost}
The distance between ${\bf v} = ((a_i,p_i))_{i \in I_1}$ and  ${\bf v'} = ((a_i,p'_i))_{i \in I_2}$
can be computed using $\mathcal{O}(|I_1| + |I_2|)$ $\lca$-queries.
\end{lemma}
\begin{proof}
We calculate their distance straight from the definition.
This can be done using $|I_1|+|I_2|$ $\lca$ queries, by going through consecutive elements of these vectors:
as they are sorted, seeing $(a_i,p_i)$ and $(a_{i'},p'_{i'})$
we can decide whether $a_i = a_{i'}$, in which case $i=i' \in I_1 \cap I_2$,
or $a_i > a_{i'}$, and thus $i \in I_1\setminus I_2$;
or $a_i < a_{i'}$, when $i' \in I_2 \setminus I_1$.
In the first case, we calculate $d(p_i,p'_{i'})$, this can be done by comparing
$m(p_i)$ and $m(p'_{i'})$ and by calculating $\lca(p_i,p'_{i'})$,
if $m(p_i) = m(p'_{i'})$.
This distance is compared with the current maximum.
If $i \in I_1 \setminus I_2$, we compare the current maximum with $m(p_i)+1$;
the situation for $p'_{i'}$ is symmetric.
%

There are at most $|I_1| \cup |I_2|$ $\lca$ queries used and at most as much
other operations, which are all performed in constant time.
So the cost of the whole procedure can be charged to the $|I_1| \cup |I_2|$ $\lca$ queries. 
\qed
\end{proof}

To show that the total time of the construction of the the fragment for $Q_t$
is $\O (|\delta_t| \log t)$, where $\delta_t$ is the transition function restricted
to the input from $Q_t$, we estimate separately
the cost of finding the ancestors of the vectors and the cost of grouping of the vectors, according to their $2^k$ successors.
However, to properly estimate the time needed for that, we cannot use the whole already constructed part
of the distance tree, as it may be vary large comparing to $Q_t$.
Thus, as a~preprocessing step, we extract out of the distance tree the distance tree induced by the states
appearing in the vectors.
Being more precise, we calculate the subtrees $\mathcal T_a$ for $a \in \Sigma$:
it is a sub-distance tree for states $\delta(Q_t,a)$.
\begin{lemma}
Constructing $\mathcal T_a$ for $a \in \Sigma$ can be done in time $\O (|\delta_t|\log |Q_t|)$.
\end{lemma}
\begin{proof}
Firstly we calculate the states in  $\delta(Q_t,a)$ for each $a$, in total time $\mathcal{O}(|\delta_t|)$.
Let us a fix a letter $a$.
We sort the leaves in $\mathcal T_a$ in time $(|\mathcal T_a| \log |Q_t|)$:
we assume that they are in some arbitrary, but fixed, order in $\mathcal D(M)$.
Let them be $q_1$, $q_2$, \ldots , $q_s$.
We build $\mathcal T_a$ by successively adding leaves.
Suppose that a~tree for $q_1$, \ldots , $q_i$ has already been built.
The right-most path (names of nodes and their levels) is kept as a list.
To add $q_{i+1}$, we calculate the $\lca(q_i,q_{i+1})$ and remove from
the list all nodes with smaller level.
If the last element of the list, call it $v$, has level greater than $\lca(q_i,q_{i+1})$
we create a~new inner node $v'$ in $\mathcal T_a$,
insert it into the right-most path as the last element,
make the right-most son of $v$ the only child of $v'$ and $v'$ the new right-most son of $v$.
Next we make $q_{i+1}$ a~right-most child of the last node in the list
(which might be $v$ or $v'$).

Since each node is inserted and removed from the right-most path at most once,
the total running time is $\mathcal{O}(|\mathcal T_a|)$.
Summing up over all $a \in \Sigma$, we obtain that the total construction time is
$\mathcal{O}(|\delta_t| )$.
\qed
\end{proof}

\begin{lemma}
\label{lem:going_up}
The total time of finding ancestors of vectors is $\mathcal{O}(|\delta_t|\log n)$.
\end{lemma}
\begin{proof}
Finding the ancestors is implemented naively: for each state $q$ in the vector
we traverse $\mathcal{D}(M)$ up $2^k$ steps up from $q$.

Consider a pair $(a,p)$ in one of the vectors $\bf v$ and one of the edges $e$ in $\mathcal T _a$
on the way from $p$ to the spine.
We show that $e$ is traversed at most $\lceil \log t \rceil$ times
when constructing the distance tree for $Q_t$.

Note, that the recursive calls are made for
$k = \lceil \log t \rceil, \lceil \log t \rceil - 1, \ldots , 1$.
So it is enough to show that $e$ is traversed once for a fixed value of $k$.

Consider a given instance of a recursive call.
Then when the sub-recursive calls are made, $e$ goes to exactly one of these sub-calls,
except when one of the $2^k$-ancestor of leaves lays on $e$.
However, in such a case there is no need to traverse $e$ by any sub-call
from the lower group: the lower end of $e$ is an ancestor of all vertices in this tree.
So we can modify the algorithm a bit:
as soon as some edge is to be traversed,
we check if the root lies on this edge.
If so, this edge is not traversed, as searched node
is implicit and therefore represented by the lower end and an offset.
And so $e$ is traversed at most once for each $k$, which concludes the proof.
\qed
\end{proof}

\begin{lemma}
\label{lem:grouping}
Given $k$, the total size of grouping vectors in the recursive calls for $k$
is linear in their size plus an additional cost, which is
$\mathcal{O}(|\delta_t|\log |Q_t|)$ summed over all $k$
\end{lemma}
\begin{proof}
We want to sort lexicographically vectors ${\bf v}_1, \ldots , {\bf v}_i$
of integers in the range $1, 2, \ldots, n + |\Sigma|$.
This can be done in a standard way (say, usinf \proc{RadixSort}) in time $\mathcal{O}(n + |\Sigma| + \sum_{j=1}^i |{\bf v}_j|)$.
This is too much, as $n$ and $|\Sigma|$ might be large compared to $\sum_{j=1}^i |{\bf v}_j|$.
Thus we can do the following. In each $\mathcal T_a$ we can replace each node by
a number from range $1, \ldots , 2|Q_t|-1$.
Then instead of sorting according to names in
$\mathcal D (M)$, we use the local names from $\mathcal T_a$.
In this way the running time is $\mathcal{O}(|Q_t| + |\Sigma| + \sum_{j=1}^i |{\bf v}_j|)$

Still, for small instances, $|Q_t| + |\Sigma|$ can be substantially larger than
$ \sum_{j=1}^i |{\bf v}_j|$.
To avoid this problem, we process all the recursive call for a fixed $k$ in parallel.
Then the sorting is done for all vectors in the recursive calls.
Note, that if vectors come from different recursive calls, they cannot have the same non-trivial ancestors
(and trivial, i.e., empty, ones can be identified and removed from the sorting beforehand).
Thus the additional cost $\mathcal{O}(|Q_t| + |\Sigma|)$ is included once for each $k=1, \ldots, \log |Q_t|$,
and so in total gives $\mathcal{O}((|Q_t| + |\Sigma|) \log |Q_t|)$ time, which is $\mathcal{O}(|\delta_t|\log |Q_t|)$.
\qed
\end{proof}

Two previous lemmata allow calculating the whole recursion time
\begin{lemma}
\label{lem:recursive calls}
The cost of the procedure for vectors ${\bf v_1}, \ldots , {\bf v_\ell}$, excluding the additional cost
of $\mathcal{O}(|\delta_t|\log |Q_t|)$ from Lemma~\ref{lem:grouping},
is $\O ((\sum_{i=1}^\ell |{\bf v_i}| - |{\bf v}|) \log |Q_t|)$,
where ${\bf v}$ is a lowest common ancestor of ${\bf v_1}, \ldots , {\bf v_\ell}$.
\end{lemma}
\begin{proof}
First note, that if $\bf v$ is the lowest ancestor of ${\bf v_1}, \ldots , {\bf v_\ell}$
and $\bf v'$ is some ancestor of these vectors, then $|{\bf v}|\geq |{\bf v'}|$,
and so the estimation using $|{\bf v'}|$ is weaker than the one using |{\bf v}|.

The claim is shown by an induction.
Fix a constant $c$, for which it is shown that the total cost is at most $c (\sum_{i=1}^\ell |{\bf v_i}| - |{\bf v}|) \log s$.

The basis of the induction are calls for at most two vectors.
No recursive call is made for one vector, and so we do not consider it.
As observed in Lemma~\ref{lem:distancecost},
when there are only two vectors $\bf v$,$\bf v' $, the calculation is done
using $\O ({\bf v} + {\bf v'})$ $\lca$ queries. So $c$ can be chosen in advance so that this
is at most $c ({\bf v} + {\bf v'})\log t$.

When there are more than two vectors,
by Lemma~\ref{lem:grouping} we can group them according to their ancestors
in time at most $\sum_{i=1}^\ell |{\bf v_i}|$.
Since $ |{\bf v}| \leq \min_{i=1}^\ell |{\bf v_i}|$,
constant $c$ can be chosen in advance in the way that this is at most $c(\sum_{i=1}^\ell |{\bf v_i}| - |{\bf v}|)$ time.

Then there are sub-calls made. Let ${\bf u_1}, \ldots , {\bf u_{\ell'}}$ be the ancestors of vectors.
Then the `upper' recursive call, by the induction assumption, takes at most
$c (\log s -1)(\sum_{i=1}^{\ell'}{\bf u_i} - {\bf v}) $
(note, that ${\bf v}$ is a common ancestor of ${\bf u_1}, \ldots , {\bf u_{\ell'}}$).
The `lower' subcalls take, in total, time $c (\sum_{i=1}^\ell |{\bf v_i}| - \sum_{i=1}^{\ell'}|{\bf u_i}|) (\log s -1)$,
since each ${\bf u_i}$ is a common ancestor of vectors for one of these subcalls.
Hence all the calls take at most $c (\sum_{i=1}^\ell |{\bf v_i}| - |{\bf v|}) (\log s -1)$ time.

Summing the cost of the recursive calls and the grouping:
\begin{equation*}
c(\sum_{i=1}^\ell |{{\bf v_i}}| - |{{\bf v}|})  + c (\sum_{i=1}^\ell |{\bf v_i}| - |{\bf v}|) (\log s -1) \leq 
c (\sum_{i=1}^\ell |{\bf v_i}| - |{\bf v}|) \log s,
\end{equation*}
as claimed.
\qed
\end{proof}

\begin{proof}[of Theorem~\ref{thm:finite distance}]
By Lemma~\ref{lem:grouping} together with Lemma~\ref{lem:recursive calls}.
\end{proof}

\section{Additional material for Section~\ref{sec:hyper}}
\begin{figure}
  \centering
  \includegraphics[width=\textwidth]{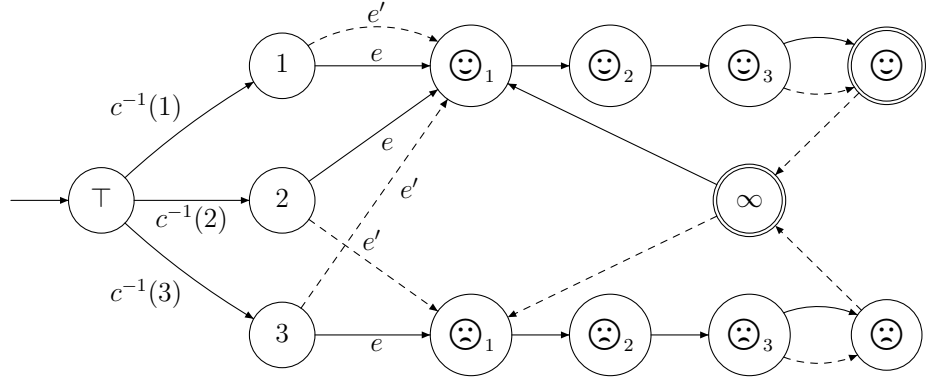}
  \caption{DFA~$c(M)$ constructed in~\protect{Section~\ref{sec:HCon}},
    where $a$-transitions are represented by unbroken lines (unless
    noted otherwise), $b$-transitions are represented by dashed lines,
    and $e = \{v_1, v_2\}$ and $e' = \{v_3, v_4\}$ with $v_1 < v_2 <
    v_3 < v_4$ and $c(v_1) = 1$, $c(v_2) = 3$, $c(v_3) = 1$, and $c(v_4)
    = 2$. The state $\bot$ is not depicted.}
  \label{fig:Hyper3}
\end{figure}

The DFA~$c(M)$ has $14$~states, and it can easily be verified that it
is hyper-equivalent to~$M$ and has no different, but equivalent kernel
states.  It is depicted in Figure~\ref{fig:Hyper3}.

\begin{proof}[of Lemma~\ref{lm:hmain}]
  Let $N = \langle P, \Sigma, \mu, p_0, F'\rangle$ be a DFA such that
  $M \sim N$ and $\abs P \leq 14$.  Without loss of generality we may
  assume that there are no different, but equivalent kernel states
  in~$N$.  Note that the DFA~$M$ is minimal provided that $G$~has no
  vertex without any incident edge.
  By~\cite[Theorem~3.8]{BadrRAIRO}\footnote{Actually, the cited
    theorem assumes~$N$ to be minimal, but the original proof also
    works in our relaxed setting (i.e., when there are no different,
    but equivalent states in the kernel).}  there exists a mapping~$h
  \colon Q \to P$ such that $q \sim h(q)$ for every $q \in Q$, which
  additionally is an isomorphism~$h \colon \Ker(M) \to \Ker(N)$
  between the kernel states.  Since $\Ker(M) = \{\bot, \infty,
  \smiley, \frowny \} \cup \{ \Circle_j \mid \Circle \in \{\smiley,
  \frowny\}, j \in [3]\}$, we have $\{\infty, \smiley\} \subseteq F'$
  and $10$~distinct states in~$N$ that behave like their counterparts
  in~$M$ [i.e., $L_M(q) = L_N(h(q))$ for every $q \in \Ker(M)$].
  Since $\top \not\sim p$ and $\delta(v) \not\sim p$ for every $p \in
  \Ker(N)$ and $v \in V$, we conclude that $h(\top) = p_0 \notin
  \Ker(N)$ and $h(\delta(v)) \notin \Ker(N)$, which means that we
  identified the $11^{\text{th}}$~and~$12^{\text{th}}$ state in~$N$
  because they are not hyper-equivalent to each other (i.e., $\top
  \not\sim \delta(v)$ for every $v \in V$).  Consequently, there are
  at most $2$~other unidentified states.

  \begin{clm}
  \label{clm:EV errors}
  The DFA~$N$ commits at least $\abs E \cdot (\abs V -2)$ errors with
  prefix~$ve$ such that $v \in V$ is a vertex and $e \in E$ is a
  non-incident edge.  Exactly $\abs E \cdot (\abs V - 2)$ such errors
  are committed if $\mu(ve) \in \{h(\smiley_1), h(\frowny_1) \}$ for
  every $v \in V$~and~$e \in E$ with $v \notin e$.
  \end{clm}

  \begin{proof}
    Let $v \in V$ be a vertex and $e \in E$ be a non-incident edge,
    which means that $v \notin e$.  Clearly, $\delta(vea^{j-1}) =
    \blacksmiley_j$ for all $j \in [3]$ and $\blacksmiley_i \not\sim
    \mathord{\blacksmiley_j}$ for all different $i, j \in [3]$.
    Consequently, also the $3$~states of~$S = \{\mu(vea^{j-1}) \mid j
    \in [3]\}$ are pairwise hyper-inequivalent (and hence different).  Since
    there are at most $2$~unidentified states, at least one state
    of~$S$ is already identified.  Moreover, $\delta(vea^{j-1}) \sim
    \mu(vea^{j-1})$ for all $j \in [3]$ yields that $\mu(vea^{j-1})
    \in \{h(\smiley_j), h(\frowny_j)\}$ for at least one $j \in [3]$,
    since $\smiley_j$ and $\frowny_j$ are the only states hyper-equivalent
    to $\smiley_j$ and $\mu(vea^{j-1}) \sim \smiley_j$.
    Since $\sim$ is a congruence and $h$ is an isomorphism
    on~$\Ker(M)$, we obtain that $\mu(vea^2) \in \{h(\smiley_3),
    h(\frowny_3)\}$ in any case.  However, $\delta(vea^2) =
    \blacksmiley_3$, so we obtain the error word~$vea^2x$ for some $x \in
    \{a, b\}$ because
    \begin{itemize}
    \item $L_N(h(\smiley_3)) = L_M(\smiley_3)$ and $L_N(h(\frowny_3))
      = L_M(\frowny_3)$,
    \item $L_M(\blacksmiley_3) \simdiff L_N(h(\smiley_3)) =
      L_M(\blacksmiley_3) \simdiff L_M(\smiley_3) = \{b\}$, and
    \item $L_M(\blacksmiley_3) \simdiff L_N(h(\frowny_3)) =
      L_M(\blacksmiley_3) \simdiff L_M(\frowny_3) = \{a\}$.
    \end{itemize}

    Moreover, if $\mu(ve) \in \{h(\smiley_1), h(\frowny_1)\}$, then
    there is exactly one error word with prefix~$ve$.  Overall, this
    yields at least one error for every $v \in V$ and $e \in E$ with
    $v \notin e$.  The total number of such errors is at least $\abs E
    \cdot (\abs V - 2)$, and it is exactly $\abs E \cdot (\abs V - 2)$
    if $\mu(ve) \in \{h(\smiley_1), h(\frowny_1)\}$ for every $v \in
    V$ and $e \in E$ with $v \notin e$.  \qed
  \end{proof}

  Let $S = \{ \mu(v) \mid v \in V\}$.  For every $v \in V$, the
  state~$\delta(v)$ is not hyper-equivalent to any state of~$\Ker(N)$
  and $\delta(v) \not\sim p_0$, which yields that $\abs S \leq 3$;
  without loosing generality we can assume that $|S|=3$,
  as this is the hardest case.
  Let $S = \{p_1, p_2, p_3\}$, and let $c \colon V \to [3]$ be such
  that for every $v \in V$ we have $c(v) = i$ if and only if $\mu(v) =
  p_i$.  Thus, we deduced a~$3$-colouring from the transitions of~$N$.
  Next, we investigate colouring violations and its connection
  with the number of errors with respect to $M$.  For every edge $e =
  \{v_1, v_2\} \in E$ such that $c(v_1) = c(v_2)$, we have
  $L_N(\mu(v_1e)) = L_N(\mu(v_2e))$, but
  \[ L_M(\delta(v_1e)) \simdiff L_M(\delta(v_2e)) = L_M(\smiley_1)
  \simdiff L_M(\frowny_1) = \{aaa, aab\} \enspace. \] Consequently, we
  obtain at least $2$~additional errors for every such edge~$e$
  because the corresponding error words start with $ve$ where $v \in
  e$.  Those two errors together with the errors described in
  Claim~\ref{clm:EV errors} yield that $N$~commits strictly more than
  $\abs E \cdot (\abs V - 2)$~errors if $G$~is not $3$-colourable.

  Finally, we show that the DFA~$c(M) = \langle P, \Sigma, \mu, \top,
  F \rangle$, which has exactly $14$~states, commits exactly $\abs E
  \cdot (\abs V - 2)$~errors, if $G$~is $3$-colourable by the proper
  $3$-colouring $c \colon V \to [3]$.  By Claim~\ref{clm:EV errors},
  $c(M)$~commits exactly $\abs E \cdot (\abs V - 2)$~errors with
  prefix~$ve$ for $v \in V$ and $e \in E$ with $v \notin e$ because
  $\mu(ve) \in \{\smiley_1, \frowny_1\}$ for all such $v$~and~$e$.
  Clearly, all remaining error words must start with $ve$ for some $v
  \in V$ and $e \in E$ such that $v \in e$.  However, if $v \in e$,
  then $\delta(ve) = \mu(c(v)e)$ by the definition of $c(M)$,
  which yields that no error word with prefix~$ve$ exists.
  \qed
\end{proof}

\section{Additional material for Section~\ref{sec:kMinHardness}}

Recall that $M = \langle Q, \Sigma,
\delta, q_0, F\rangle$ is our DFA and that we can always rename states to
avoid conflicts.  In particular, we assume that $Q \cap \nat =
\emptyset$.

\begin{lemma}[full version of Lemma~\ref{lm:kCong}]
  \label{lm:kCongII}
  For every congruence~$\mathord{\simeq} \subseteq Q \times Q$ on~$M$,
  there exists a DFA~$N = \langle P, \Delta, \mu, q_0, F' \rangle$ such that
  \begin{enumerate}
  \item \label{en:gadget preserves 1}
  $L_M(q) = L_N(q) \cap \Sigma^*$ and $\delta^{-1}(q) =
    \mu^{-1}(q)$ for every $q \in Q$,
  \item $L_M(q_1) \simdiff L_M(q_2) = L_N(q_1) \simdiff L_N(q_2)$ for
    all $q_1 \simeq q_2$,
  \item $p_1 \not\sim p_2$ for every $p_1 \in P \setminus Q$ and
    $p_2 \in P$ with $p_1 \neq p_2$, and
  \item $q_1 \not\sim q_2$ in~$N$ for all $q_1 \not\simeq q_2$.
  \end{enumerate}
\end{lemma}

\begin{proof}
  Let $n$~be the index of~$\simeq$ and $\theta \colon
  (Q/\mathord{\simeq}) \to [n]$ be an arbitrary bijection.  Moreover,
  let
  \begin{itemize}
  \item $P = Q \cup [n]$ and $F' = F \cup \{n\}$,
  \item $\Delta = \Sigma \cup \{t, x\}$ where $t,x \notin \Sigma$ are
    different, new letters,
  \item $\mu(q, \sigma) = \delta(q, \sigma)$ for every $q \in Q$
    and $\sigma \in \Sigma$,
  \item $\mu(q, t) = \theta([q]_{\mathord{\simeq}})$ for every
    $q \in Q$, and
  \item $\mu(i, x) = i + 1$ for every $i \in [n-1]$ and
    $\mu(n, x) = 1$.
  \end{itemize}
  Next, we discuss hyper-equivalence and $k$-similarity in the
  DFA~$N$.  Clearly, $i \not\sim j$ for all different $i, j \in [n]$.
  Since $\sim$~is a congruence, we obtain that $q_1 \not\sim q_2$ for
  all $q_1, q_2 \in Q$ such that $q_1 \not\simeq q_2$.
  Moreover, $i \not\sim
  q$ for every $i \in [n]$ and $q \in Q$.  The~(i)
  is obvious.  Finally, we can prove by an induction on the length of~$w
  \in \Delta^*$ that $w \in L_M(q_1) \simdiff L_M(q_2)$ if and only if
  $w \in L_N(q_1) \simdiff L_N(q_2)$ for all $q_1 \simeq q_2$.  \qed
\end{proof}

In the future, we will use the construction in the proof of
Lemma~\ref{lm:kCongII} as a~gadget with $n$~states and simply assume
that we can enforce that $q_1 \not\sim q_2$ for every $q_1
\not\simeq q_2$ and every congruence~$\simeq$ with index~$n$.  This
can be done since none of the newly added states is \mbox{$k$-similar}
to an existing state and all newly added states are pairwise
dissimilar.  Moreover, since we will only merge $k$-similar states and
$\mu(q_1, t) = \mu(q_2, t)$ for all $q_1 \sim_k q_2$, the
gadget does not introduce new errors.

\begin{lemma}[full version of Lemma~\ref{lm:kInlevel}]
  \label{lm:kInlevelII}
  For every subset~$S \subseteq Q \setminus \{q_0\}$ of states and
  mapping~$\mathord{\minlevel} \colon S \to \nat$, there is a DFA~$N
  = \langle Q \cup I, \Sigma \cup \Delta, \mu, q_0, F \rangle$ such
  that
  \begin{itemize}
  \item $L_M(q) = L_N(q)$ for every $q \in Q \setminus \{q_0\}$,
  \item $\abs{\mu^{-1}(i)} = 1$ for every $i \in I$,
  \item $\abs{\delta^{-1}(s) \simdiff \mu^{-1}(s)} = 1$ for $s
    \in S$ with $\minlevel(s) \geq 2$ and $\delta^{-1}(q) = \mu^{-1}(q)$
    for all remaining $q \in Q$, and
  \item $\inlevel_N(s) \geq \minlevel(s)$ for every $s \in S$.
  \end{itemize}
\end{lemma}

\begin{proof}
  Let $n = \max {} \{ \minlevel(s) \mid s \in S\}$~be the maximal
  requested level.  We construct the DFA~$N$ such that
  \begin{itemize}
  \item $I = [n]$ (supposing that $Q \cap \nat = \emptyset$),
  \item $\Delta = \{d\} \cup \{d_s \mid s \in S\}$ are new, different
    letters,
  \item $\mu(q, \sigma) = \delta(q, \sigma)$ for every $q \in Q$
    and $\sigma \in \Sigma$, and
  \item $\mu(q_0, d) = 1$ and $\mu(i, d) = i+1$ for every $i \in
    [n-1]$,
  \item $\mu(\minlevel(s)-1, d_s) = s$ for every $s \in S$ such that
    $\minlevel(s) \geq 2$.
  \end{itemize}
  Clearly, $L_M(q) = L_N(q)$ for every $q \in Q \setminus \{q_0\}$
  and $\inlevel_N(s) \geq \minlevel(s)$ for every $s \in S$.  Finally,
  $\mu^{-1}(i) = \{d^i\}$ for every $i \in [n]$, which can be used to
  prove the remaining statements.  \qed
\end{proof}

We can use the first gadget to make sure that all newly introduced
states in the previous construction are $k$-dissimilar.  Mind that a
renaming can be used to ensure that $Q \cap \nat = \emptyset$.


We use the gadgets of Lemmata~\ref{lm:kCongII} and~\ref{lm:kInlevelII}
as follows: We select the states $S = \{1_0, 2_0, 3_0, 1_1, 2_1,
\smiley_0, \frowny \}$ and let $\mathord{\minlevel} \colon S \to \nat$ be
such that
\begin{alignat*}{6}
  && \minlevel(1_0) &= \minlevel(2_0) &{}= \minlevel(3_0) &= 3s - 1\\
  & \minlevel(\smiley_0) &{}= \minlevel(\frowny) &= \minlevel(1_1) &{}=
  \minlevel(2_1) &= k + 1 \enspace.
\end{alignat*}
Moreover, let $\mathord{\simeq} \subseteq Q \times Q$ be the
equivalence induced by the partition (single-element
classes are omitted)
\begin{align*} &\{\bot, \blacksmiley, \frowny\} \cup \{\smiley_i \mid
  0 \leq i \leq s\} \cup \{i_j \mid i \in [2], j \in [\ell]\} && \{v
  \mid v \in V\} \cup \{1_0, 2_0, 3_0 \} \enspace.
\end{align*}
It can easily be checked that $\mathord{\simeq}$~is a congruence.  Let
$M = \langle Q, \Sigma, \delta, 0, F \rangle$ be the DFA obtained
after adding the gadgets of Lemmata~\ref{lm:kCongII}
and~\ref{lm:kInlevelII} using the above parameters.  The obtained
DFA~$M$ is illustrated in Fig.~\ref{fig:main_states}.  We state some
simple properties that follow immediately from the gadgets.

\begin{lemma}
  \label{lm:kTriv}
  We observe the following simple properties:
  \begin{itemize}
  \item Hyper-inequivalence is as indicated in Fig.~\ref{fig:main_states};
    i.e., $q_1 \not\sim q_2$ for every $q_1 \not\simeq q_2$.
  \item The levels are as indicated in
    Fig.~\ref{fig:main_states}.  More specifically,
    for $j \in [\ell]$, $i \in [s-1]$, and $v \in V$.
    \begin{align*}
      \inlevel_M(1_0) &= \inlevel_M(2_0) = 3s - 1 &
      \inlevel_M(\bot) &= \infty \\
      \inlevel_M(3_0) &= 3s - 1 & \\
      \inlevel_M(1_j) &= \inlevel_M(2_j) = k+j & \inlevel_M(v) &= s+1 \\[1ex]
       \inlevel_M(\smiley_i) &= k+i+1 & \inlevel_M(\smiley_s) &= k +
       \ell + 1 \\
       \inlevel_M(\blacksmiley) &= 3s & \inlevel_M(i) &= i \\
       \inlevel_M(\frowny) &= k+1 & \inlevel_M(s) &= s
    \end{align*}    
  \end{itemize}
\end{lemma}

\begin{proof}
  Both properties follow immediately from Lemmata
  \ref{lm:kCongII}~and~\ref{lm:kInlevelII} and trivial inductions.  \qed 
\end{proof}


\begin{lemma}
  \label{lm:kB}
  \mbox{ }
  \begin{enumerate}
  \item The difference between the right-languages of the states
    $\{1_0, 2_0, 3_0\}$ is small.  More precisely,
    \begin{align*}
      L_M(1_0) \simdiff L_M(2_0) &= \{ab^{\ell-1} \} \\
      L_M(1_0) \simdiff L_M(3_0) &= \{ab^{\ell-1}, ab^\ell \} \\
      L_M(2_0) \simdiff L_M(3_0) &= \{ab^\ell \} \enspace. 
    \end{align*}
  \item Additionally, $\abs{\delta^{-1}(v)} = 2^s$ for
    every $v \in V$, and 
    \begin{align*}
      \abs{L_M(\smiley_i) \simdiff L_M(\frowny)} &= 2^{s-i} &&&
      d(\smiley_i, \frowny) &= s-i+1
      \\
      \abs{L_M(\smiley_i) \simdiff L_M(\blacksmiley)} &\geq 2^{s-1}
      &&& d(\smiley_i, \blacksmiley) &= s+1 \\     
      \abs{L_M(\blacksmiley) \simdiff L_M(\frowny)} &= 2^{s-1} &&&
      d(\blacksmiley, \frowny) &= s+1
      \\
      \abs{L_M(1_j) \simdiff L_M(\blacksmiley)} \geq
      \abs{L_M(\blacksmiley)} &= 2^{s-1} &&&
      d(1_j, \blacksmiley) &= \max(s,\ell-j+1)+1 \\
      \abs{L_M(2_j) \simdiff L_M(\blacksmiley)} \geq
      \abs{L_M(\blacksmiley)} &= 2^{s-1} &&& d(2_j, \blacksmiley)
      &= \max(s,\ell-j+1)+1
    \end{align*}
    for every $0 \leq i \leq s$ and $j \in [\ell]$.
  \item Finally, $\mathord{\sim_k}$ is the reflexive and
    symmetric closure of
    \begin{align*}
      &\phantom{{}\cup{}} \mathord{\equiv} \cup \{ (i_0, v) \mid i \in
      [3], v \in V \} \cup
      \{ (v_1, v_2) \mid v_1, v_2 \in V\} \cup {} \\
      &{}\cup \{ (\blacksmiley, \smiley_i) \mid 0 \leq i \leq s\} \cup
      \{(\blacksmiley, \frowny), (\blacksmiley, \bot) \} \cup
      \{(\blacksmiley, i_j) \mid i \in [2], s + 1 < j \leq
      \ell \}.
      \end{align*}
  \end{enumerate}
\end{lemma}

\begin{proof}
  Properties (i)~and~(ii) can be observed easily.  We turn to (iii):
  let $i, i' \in [3]$ such that $i \neq i'$.  Since $\inlevel_M(i_0) =
  \inlevel_M(i'_0) = 3s-1$ by Lemma~\ref{lm:kTriv} and $d(i_1, i'_1)
  \geq \ell-1$ by (i), we obtain $d(i_0, i'_0) + 3s \geq (k -
  2s) - 1 + 3s = k + s - 1 \geq k$, which proves that $i_0 \not\sim_k
  i'_0$.  Similarly, we can compute $d(i_0, v) + s + 1 \leq$ $(k - 2s + 1) +
  s + 2 =$ $ k - s + 3 \leq k$ for every $v \in V$, which proves $i_0 \sim_k
  v$.  Now, let $v_1, v_2 \in V$.  Then
  \[
  \inlevel_M(v_1) = \inlevel_M(v_2) = s+1 \qquad \text{and} \qquad
  d(v_1, v_2) \leq s + 2 \enspace,
  \]
  which yields $2s+3 < k$ and
  proves that $v_1 \sim_k v_2$.  Since $\inlevel(\blacksmiley) = 3s$
  and $d(\smiley_0, \blacksmiley) = s+1$, we obtain $3s + s + 1 = 4s +1 \leq k$,
  which proves that $\smiley_0 \sim_k \blacksmiley$.  In essentially
  the same way, we can prove all similarities to~$\blacksmiley$, which
  yields that we proved all similarities.  Clearly, two states $q_1,
  q_2 \in Q$ such that $\min(\inlevel_M(q_1), \inlevel_M(q_2)) \geq k$
  are $k$-similar if and only if $q_1 \equiv q_2$, which proves the
  nontrivial dissimilarities.  \qed
\end{proof}

Using a proper $3$-colouring $c \colon V \to [3]$ we define the DFA $c(m)$.

\begin{definition}
  \label{df:kOut}
  Let $c \colon V \to [3]$ be a $3$-colouring and $c(M) = \langle P,
  \Sigma, \mu, 0, F\rangle$ be the DFA such that
  \begin{itemize}
  \item $P = \{\bot\} \cup \{i_j \mid i \in [2], j \in [\ell]\} \cup
    [0, s] \cup \{\smiley_i \mid 0 \leq i \leq s\}$
  \item for every $v \in V$, $e = \{v_1, v_2\} \in E$ with $v \notin
    e$ and $v_1 < v_2$, $i \in [s]$, $j \in [\ell]$, and $j' \in [3]$
    \begin{align*}
      \mu(i-1, a) &= i & \mu(\smiley_{i-1}, a) &= \smiley_i &
      \mu(1_{j-1}, a) &= 1_j \\
      \mu(i-1, b) &= i & \mu(\smiley_{i-1}, b) &= \smiley_i
      & \mu(2_{j-1}, a) &= 2_j \\
      \mu(s, v) &= c(v)_0 & \mu(1_\ell, a) &= \smiley_s & \mu(2_\ell,
      a) &= \smiley_s \\ 
      \mu(j'_0, e) &= \begin{cases}
        \smiley_0 & \text{, if } c(v_2) \neq j' \\
        \bot & \text{, otherwise.}
      \end{cases}
    \end{align*}
  \item For all remaining cases, we set $\mu(q, \sigma) = \bot$.
  \end{itemize}
\end{definition}

The DFA $c(M)$ is illustrated in Figure~\ref{fig:Hyper4}.
We first show that $c(M)$ is $k$-similar to $M$ and then that
it is $k$-minimal.

\begin{lemma}
  \label{lm:kYes}
  The constructed DFA~$c(M) = \langle P, \Sigma, \mu, 0, F \rangle$ is
  $k$-similar to~$M$.
\end{lemma}

\begin{proof}
  Let $\Delta$ be the alphabet of letters without the letters used by
  the gadgets, which are collected in~$\Gamma$.  We first observe the
  equalities
  \begin{align*}
    L_N(\smiley_s) &= L_M(\smiley_s) & L_N(\smiley_s) \cap \Delta^* &=
    \{\varepsilon\} \\
    L_N(\bot) &= L_M(\bot) = L_M(\frowny) & L_N(\bot) \cap \Delta^* &=
    \emptyset \\
    L_N(1_j) &= L_M(1_j) & L_N(1_j) \cap \Delta^* &= \{ b^{\ell - j},
    b^{\ell - j + 1} \} \\
    L_N(2_j) &= L_M(2_j) & L_N(2_j) \cap \Delta^*&= \{ b^{\ell - j +
      1} \} \\[1ex]
    L_N(\smiley_i) &= L_M(\smiley_i) & L_N(\smiley_i) \cap \Delta^* &=
    \{a, b\}^{s - i} \\
    \mu^{-1}(i) &= \delta^{-1}(i) & \mu^{-1}(i) &= \{a, b\}^i \enspace.
  \end{align*}
  In addition, since $\{ i_j \mid i \in [2], j \in [\ell]\}$ forms a
  single class of~$\simeq$ in both $c(M)$~and~$M$ it holds that
  \[ L_{c(M)}(i_j) \cap \Delta^*\Gamma\Sigma^* = L_M (i'_{j'}) \cap
  \Delta^*\Gamma\Sigma^* \] for all $i, i' \in [2]$ and $j, j' \in
  [\ell]$.  From those statements and simple 
  applications of the
  statements of Lemmas \ref{lm:kTriv}~and~\ref{lm:kB} we can easily
  conclude that $M$~and~$c(M)$ are $k$-similar.  \qed
\end{proof}

\begin{figure}
  \centering
  \includegraphics[scale=0.7]{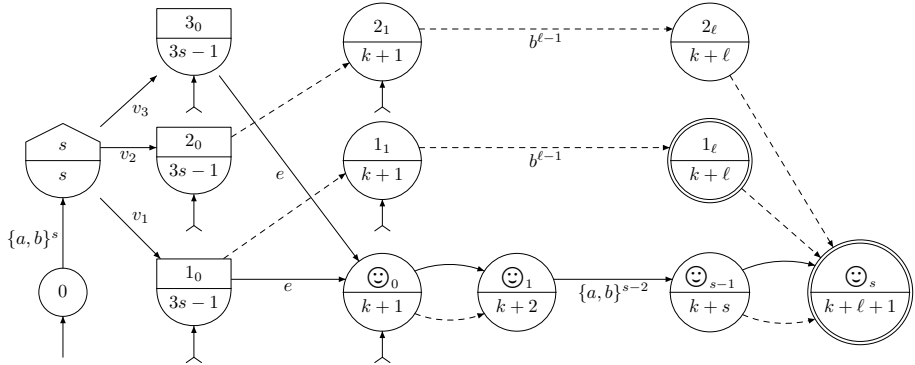}
  \caption{Illustration of the DFA~$c(M)$ of
    Section~\protect{\ref{sec:kCons}}}
  \label{fig:Hyper4}
\end{figure}

Note, that all states in $c(M)$ are pairwise
$k$-dissimilar by Lemma~\ref{lm:kB}.  Consequently, they form a
maximal set of pairwise $k$-dissimilar states in~$M$, which yields
that~$\abs P$ coincides with the number of states of all $k$-minimal
DFA for~$M$ by Lemma~\ref{cor:kSet}. Thus DFA $M$ is $k$-minimal.

We finally show that our construction is correct.

\begin{proof}[of Lemma~\ref{lm:kMain}]
  Let $N = (P, \Sigma, p_0, \mu, F')$ be a DFA that is $k$-minimal
  for~$M$.  We select a maximal set~$Q'$ of pairwise $k$-dissimilar
  states by
  \[ Q' = \{\bot\} \cup \{i_j \mid i \in [2], j \in [\ell]\} \cup [0,
  s] \cup \{\smiley_i \mid 0 \leq i \leq s\} \enspace. \] Let $h
  \colon Q' \to P$ be the bijection of Corollary~\ref{cor:multisets}
  using the maximal set~$Q'$.  Let \[ S = \{ \bot \} \cup \{ \smiley_i
  \mid 0 \leq i \leq s \} \cup \{ i_j \mid i \in [2], j \in [\ell]\}
  \enspace. \] Consequently, $q \equiv h(q)$ for every $q \in S$.
  From $M \sim_k N$, which yields $M \sim N$, we can conclude that
  $\delta(w) \sim \mu(w)$ for every $w \in \Sigma^*$.  Let $w \in \{a,
  b\}^i$ for some $0 \leq i \leq s$.  Then $\mu(w) = h(i)$ because
  $\mu(w) \sim \delta(w) = i$ and $h(i)$~is the only state~$p \in P$
  such that $p \sim i$.  Now, let $w = uv$ with $u \in \{a, b\}^s$ and
  $v \in V$.  Since $\delta(uv) = v$ and $v \sim i_0$ only for~$i \in
  [3]$, we obtain $\mu(uv) \in \{h(1_0), h(2_0), h(3_0)\}$.  From this
  behaviour we deduce a colouring $c \colon V \to [3]$ by $c(v) = i$
  if and only if $\mu(uv) = h(i_0)$ for every $v \in V$.  Note that
  this definition does not depend on the choice of~$u$ because $\mu(u)
  = \mu(u') = h(s)$ for all $u, u' \in \{a, b\}^s$.

  \begin{clm}
    \label{clm:kErrors}
    The DFA~$N$ commits at least
    \[ m = 2^{2s-1} \cdot \abs E \cdot (\abs V - 2) + 3 \cdot 2^{s-1}
    \cdot \abs E \] errors with prefix
    \begin{itemize}
    \item $uve$ where $u \in \{a, b\}^*$, $v \in V$ is a vertex, and
      $e \in E$ such that $v \notin e$, or
    \item $dwe$ where $d$~is the main symbol of the level gadget, $w
      \in \Sigma^*$, and $e \in E$.
    \end{itemize}
    Exactly $m$~such errors are committed if the following two
    conditions, called~$(\dagger)$, are fulfilled:
    \begin{enumerate}
    \item $\mu(uve) \in \{h(\smiley_1), h(\bot) \}$ for every $u \in
      \{a, b\}^*$, $v \in V$~and~$e \in E$ with $v \notin e$.
    \item $\mu(dwe) \in \{h(\smiley_1), h(\bot) \}$ for every $w \in
      \Sigma^*$~and~$e \in E$.
    \end{enumerate}
  \end{clm}

  \begin{proof}
    We distinguish two cases: whether the error word goes through $i_0$
    for some $i \in [3]$ or it goes through $v$ for some $v \in V$:
    \begin{itemize}
    \item The only way to arrive at~$i_0$ in~$M$ is via the level gadget.  Then by
      Lemma~\ref{lm:kInlevelII} there exists exactly one $w \in
      \delta^{-1}(i_0)$ such that $\abs w = 3s - 1$ and $w$~does not
      end with $v \in V$.  In addition, $\delta(i_0, e) =
      \blacksmiley$ for every $e \in E$.  However, in~$N$ we have
      $\mu(w) \in \{h(1_0), h(2_0), h(3_0)\}$ and $\mu(h(j_0), e) \in
      h(S)$ for every $j \in [3]$ because these are the only
      hyper-equivalent states.  Since $q \equiv h(q)$ for every $q \in
      S$, we have
      \begin{align*}
        \abs{L_N(h(q)) \simdiff L_M(\blacksmiley)} &\geq 2^{s-1} \\
        \abs{L_N(h(\bot)) \simdiff L_M(\blacksmiley)} =
        \abs{L_N(h(\smiley_0)) \simdiff L_M(\blacksmiley)} &= 2^{s-1}
      \end{align*}
      for every $q \in S$ by Lemma~\ref{lm:kB}(ii).  This yields at
      least $2^{s-1}$~error words that start with~$we$, and exactly
      $2^{s-1}$~such error words under condition~$(\dagger)$.
      Consequently, there are at least $3 \cdot 2^{s-1} \cdot \abs
      E$~such error words in total, and exactly that many under
      condition~$(\dagger)$, because the level gadget is reproduced
      exactly.
    \item Next, we consider a word $w = uve$ with $u \in \{a, b\}^s$,
      $v \in V$, and $e \in E$ such that $v \notin e$.  Clearly,
      $\delta(w) = \blacksmiley$, but $\mu(w) \in h(S)$ because~$h(S)$
      are the only states of~$P$ that are hyper-equivalent
      to~$\blacksmiley$.  Since $q \equiv h(q)$ for every $q \in S$
      and $\mu(u) = h(s)$ independently of~$u$, there are $2^s$~such
      words~$u$ and $\abs{L_M(\blacksmiley) \simdiff L_N(h(q))} \geq
      2^{s-1}$ for every $q \in S$ by Lemma~\ref{lm:kB}(ii) because
      $L_N(h(q)) = L_M(q)$.  The same lemma also allows us to conclude
      that $\abs{L_M(\blacksmiley) \simdiff L_N(h(q))} = 2^{s-1}$
      under condition~$(\dagger)$.  Overall, there are $2^s \cdot
      2^{s-1} \cdot \abs E \cdot (\abs V - 2)$ such error words, and
      exactly that many under condition~$(\dagger)$.
    \end{itemize}
    This proves the claim. \qed
  \end{proof}

  Now we investigate the number of errors introduced for a colouring
  violation in the \mbox{$3$-colouring}~$c$.  Let $e = \{v_1, v_2\}
  \in E$ be a violating edge, i.e., $c(v_1) = c(v_2)$.  Then we have
  $\mu(uv_1e) = \mu(uv_2e) \in h(S)$ for every $u \in \{a, b\}^s$.
  However, $\{ \delta(uv_1e), \delta(uv_2e) \} = \{\smiley_0, \frowny
  \}$, which yields at least $2^s \cdot 2^s$ errors by
  Lemma~\ref{lm:kB}(ii).  Since $s > \log_2(\abs V) +2$, we have that
  $2^{2s} > 2^{s+1} \cdot \abs V$, which, together with Claim~\ref{clm:kErrors}
  proves our statement if the graph is not $3$-colourable.

  Finally, let us consider the errors of~$c(M) = \langle P, \Sigma,
  \mu, 0, F \rangle$ provided that $c \colon V \to [3]$ is a proper
  $3$-colouring for~$G$.  Recall the properties mentioned in the proof
  of Lemma~\ref{lm:kYes}.  Since $c(M)$ fulfils property~$(\dagger)$,
  we already identified exactly the errors of Claim~\ref{clm:kErrors}.
  Clearly, all error words pass a state~$h(i_0)$ in $N$ with~$i \in [3]$, which
  yields that we only have to consider error words with prefix
  $uve$~or~$uva$ where $u \in \{a, b\}^*$, $v \in V$, and $e \in E$
  such that $v \in e$.
  \begin{itemize}
  \item We start with the prefix~$uva$.  As already observed we have
    that $\mu^{-1}(s) =$ $\{a, b\}^s =$ $\delta^{-1}(s)$.  Moreover,
    $\mu(s, v) = i_0$ for all $v \in c^{-1}(i)$, whereas we have $\delta(s, v)
    = v$.  Thus, we consider $L_M(v) \simdiff L_{c(M)}(i_0)$.  Since
    $\delta(v, a) = 1_1$ and $\mu(i_0, a) = i_1$, we obtain the
    potential errors
    \[ \{uvab^j \mid j \in \{\ell -1, \ell\}, u \in \{a, b\}^s \}
    \enspace. \] Consequently, we have $2^{s+1} \cdot \abs V$
    potential errors, all of which are of length at most~$s + 2 + \ell
    = k - s + 2 < k$.
  \item Finally, we have to consider the prefixes~$uve$.  If $v \in
    e$, then $\mu(i_0, e) \equiv \delta(v, e)$ by the construction
    of $c(M)$ and so no errors are introduced.
  \end{itemize}
  Summing up all identified error words, we obtain that the DFA~$c(M)$
  commits at most
  \[ 2^{2s-1} \cdot \abs E \cdot (\abs V - 2) + 3 \cdot 2^{s-1} \cdot
  \abs E + 2^{s+1} \cdot \abs V
  \] errors, all of which are of length smaller than~$k$, which also
  proves that the DFA~$c(M)$ is $k$-similar to~$M$.  Moreover, since all
  of its states are pairwise $k$-dissimilar, it is also $k$-minimal.
  \qed
\end{proof}

\end{document}